\documentclass{article}

\usepackage{arxiv}
\usepackage[utf8]{inputenc} 
\usepackage[T1]{fontenc}    
\usepackage{hyperref}       
\usepackage{url}            
\usepackage{booktabs}       
\usepackage{amsfonts}       
\usepackage{nicefrac}       
\usepackage{microtype}      
\usepackage{natbib}
\usepackage{doi}
\usepackage[utf8]{inputenc}
\usepackage{amsmath}
\usepackage{amsthm}
\usepackage{amsfonts} 
\usepackage{graphicx}
\usepackage{xcolor}
\usepackage{float}
\usepackage{algorithmicx}
\usepackage{algorithm}
\usepackage{algpseudocode}
\usepackage{mathtools}
\usepackage[title]{appendix}
\usepackage[english]{babel}
\addto\captionsenglish{\renewcommand{\figurename}{Fig.}}

\newtheorem{theorem}{Theorem}
\newtheorem{lemma}[theorem]{Lemma}
\newtheorem*{remark}{Remark}
\newtheorem{example}{Example}
\newcommand\norm[1]{\left\lVert#1\right\rVert}
\newcommand{\matr}[1]{\mathbf{#1}}
\renewcommand{\vec}[1]{\boldsymbol{#1}}
\DeclareMathOperator{\clr}{clr}
\DeclareMathOperator{\ilr}{ilr}
\DeclareMathOperator{\gilr}{gilr}
\DeclareMathOperator{\diag}{diag}
\DeclareMathOperator{\tr}{tr}
\DeclareMathOperator*{\argmax}{arg\,max}
\DeclareMathOperator*{\argmin}{arg\,min}

\title{ Extending compositional data analysis from a graph signal processing perspective}

\author{ Christopher Rieser \thanks{This research was supported by the Austrian Science Fund (FWF) under the grant number P 32819 Einzelprojekte.}  \\
	Institute of Statistics and Mathematical Methods in Economics\\
	TU Wien\\
    Wiedner Hauptstraße, 1040 Vienna, Austria  \\
	\texttt{christopher.rieser@tuwien.ac.at} \\
	\And
	Peter Filzmoser \\
	Institute of Statistics and Mathematical Methods in Economics\\
	TU Wien\\
    Wiedner Hauptstraße, 1040 Vienna, Austria  \\
	\texttt{peter.filzmoser@tuwien.ac.at} \\
}

\renewcommand{\shorttitle}{ }

\begin{document}
\maketitle

\begin{abstract}
Traditional methods for the analysis 
of compositional data consider the log-ratios between
all different pairs of variables with equal weight, typically in the
form of aggregated contributions. This is not meaningful in contexts 
where it is known that a relationship only exists between very 
specific variables (e.g.~for metabolomic pathways), while for other
pairs a relationship does not exist. Modeling absence or presence of relationships 
is done in graph theory, where the vertices represent the variables, and the
connections refer to relations. This paper links compositional data analysis
with graph signal processing, and it extends the Aitchison geometry to a setting where
only selected log-ratios can be considered. The presented framework retains the 
desirable properties of scale invariance and compositional coherence. An additional extension to include absolute information is readily made. 
Examples from bioinformatics and geochemistry underline the usefulness of this
approach in comparison to standard methods for compositional data analysis.
\end{abstract}

\keywords{Compositional Data, Log-Ratio Analysis, Graph Theory, Graph Signal Processing, Graph Laplacian}

\section{Introduction}

Since the fundamental work of \citet{aitchison1982statistical}, the field of compositional data analysis (CoDa) has received a lot of attention. 
Analyzing positive multivariate data from a compositional point of view shifts the focus from the Euclidean perspective of absolute quantities to the view of relative information. Many data sets have since then been recognized to be of compositional nature. Examples include Microbiome data \citet{gloor2017microbiome}, omics data \citet{quinn2019field}, time-use data \citet{dumuid2019compositional}, economical data and many more. Good overviews of standard theory on CoDa can be found in \citet{Aitchison86}, \citet{pawlowsky2006compositional} or \citet{filzmoser2018applied}.

In the following we will denote the space of multivariate positive real values $\{ (x_1,\ldots,x_D )' \in \mathbb{R}^D \mid x_j > 0 \, \forall j=1,\ldots,D \}$ by $\mathbb{R}^D_{+}$ and define the D-part simplex as
\begin{align*}
    \mathcal{S}^D:=\bigg\{ (x_1,\ldots,x_D)' \in \mathbb{R}^D_{+} \biggl\vert \sum_{j = 1}^D x_j =1 \bigg\} \subset \mathbb{R}^D_+ \ .
\end{align*}
For two compositions $\vec{x}=(x_1,\ldots,x_D)',\vec{y}=(y_1,\ldots,y_D)' \in \mathbb{R}^D_{+}$ and $\alpha \in \mathbb{R}$ the following two operations called perturbation and powering are defined,
\begin{itemize}
    \item 
    $\vec{x} \oplus \vec{y} :=(x_1 y_1,\ldots,x_D y_D)'$
    \item $\alpha \odot \vec{x} :=  (x_1^\alpha,\ldots,x_D^\alpha)'$,
\end{itemize}
as well as their scaled versions 
\begin{itemize}
    \item 
    $\vec{x} \oplus_{\mathcal{A}} \vec{y} := \frac{1}{\sum_{j =1}^D x_j y_j}\vec{x} \oplus\vec{y}$
    \item $\alpha \odot_{\mathcal{A}} \vec{x} := \frac{1}{\sum_{j =1}^D x_j^\alpha} \alpha \odot \vec{x}$,
\end{itemize}
see \citet{Aitchison86}. The D-part simplex $\mathcal{S}^D$ is equipped with an inner product, known as the Aitchison inner product,
 \begin{align}
        {\langle \vec{x},\vec{y} \rangle}_{\mathcal{A}} := \frac{1}{2D} \sum^D_{i,j = 1} \log\bigg(\frac{x_i}{x_j}\bigg)\log\bigg(\frac{y_i}{y_j}\bigg) , \label{aitchisonInner}
\end{align}
such that $(\mathcal{S}^D, {\langle \cdot , \cdot \rangle}_{\mathcal{A}} , \oplus_{\mathcal{A}}, \odot_{\mathcal{A}})$ is a Hilbert space with neutral element $\frac{1}{D}(1,\ldots,1)' \in \mathbb{R}^D_{+}$ and norm $\norm{\vec{x}}_{\mathcal{A}} := \sqrt{{\langle \vec{x},\vec{x} \rangle}_{\mathcal{A}}}$, see \citet{pawlowsky2015modeling}.

A common tool for the analysis of compositional data is the 
clr (centered log-ratio)-map
\begin{align} \label{standardCLR}
  \clr : \mathcal{S}^D \rightarrow \mathbb{R}^D, \quad    \clr{(\vec{x})} := \Bigg(\log\Bigg(\frac{x_1}{\sqrt[D]{\prod^D_{j=1} x_j}}\Bigg),\ldots ,\log\Bigg(\frac{x_D}{\sqrt[D]{\prod^D_{j=1} x_j}}\Bigg)\Bigg)' ,
\end{align}
which can be shown to be distance preserving on $\mathcal{S}^D$ \citet{Aitchison86}. Further, the clr-map has the following properties,
\begin{align}
    & \clr(\vec{x} \oplus_{\mathcal{A}} \vec{y}) = \clr(\vec{x}) + \clr(\vec{y}) \label{addIso}  \\ 
    & \clr(\alpha \odot_{\mathcal{A}} \vec{x}) = \alpha \clr(\vec{x}) \label{multIso} \\ 
    & {\langle \vec{x},\vec{y} \rangle}_{\mathcal{A}} = {\langle \clr(\vec{x}),\clr(\vec{y}) \rangle}_{2}  \label{innerIso} 
\end{align}
where ${\langle \cdot , \cdot \rangle}_{2}$ denotes the standard inner product in $\mathbb{R}^D$ \citet{Aitchison86}.
As the clr-map is not one-to-one onto $\mathbb{R}^D$, a modification has been considered by 
\citet{egozcue2003isometric}, called the  
ilr (isometric log-ratio)-map
\begin{align} \label{ilrCoda}
   \ilr_{\matr{V}} : \mathcal{S}^D \rightarrow \mathbb{R}^{D-1}, \quad  \ilr_\matr{V}(\vec{x}) := \matr{V}'\clr(\vec{x}) \ ,
\end{align}
where  $\matr{V} \in \mathbb{R}^{D \times (D-1)}$ is a matrix with orthogonal columns spanning the $D-1$ dimensional subspace $\{ \vec{z} \in \mathbb{R}^D \mid \sum_{j=1}^D z_j = 0 \} \subset \mathbb{R}^D$. The ilr-map is not unique depending on the chosen basis, but is an isometric one-to-one map onto $\mathbb{R}^{D-1}$ that fulfills (\ref{addIso}), (\ref{multIso}) and (\ref{innerIso}).
Thus, the purpose of a clr (ilr)-map is to transform compositional data to the 
standard Euclidean geometry for which classical tools in statistical data analysis
are appropriate and designed. 
Note that the $l$-th component of the clr-map can be written in terms of pairwise log-ratios, since
$$
\log\Bigg(\frac{x_l}{\sqrt[D]{\prod^D_{j=1} x_j}}\Bigg) =
\frac{1}{D}\Bigg( \log\frac{x_l}{x_1} + \ldots +\log \frac{x_l}{x_{l-1}}+
\log \frac{x_l}{x_{l+1}}+ \ldots +\log \frac{x_l}{x_D}\Bigg) ,
$$
for $l\in \{1,\ldots ,D\}$, and thus clr (ilr)-maps consider the information
of all pairwise log-ratios, and they receive equal weight in the analysis.
This is not always desirable, because log-ratios of pairs
which are not in a meaningful relationship might not be considered at all
for the analysis.

In this paper we investigate the connections between CoDa and signal processing on graphs, and it will be shown that CoDa can be viewed as calculus on finite graphs. The goal is to extend tools of the latter by defining a scale invariant inner product that depends on the graphical structure. Subsequently, a scale invariant isometric one-to-one map shall be identified that depends on the graph structure and on weights, such that after transforming the original data one can work 
again in a Euclidean space. 
Additionally, we will show that we obtain a framework in which also the absolute information of the data can be considered. The main idea of the paper is thus to modify the Aitchison inner product in such a way that only certain log-ratios influence the geometry of the space. This can be done by looking at a weighted version of the Aitchison inner product (\ref{aitchisonInner}), with weights $w_{ij}$:
\begin{align*}
   \frac{1}{2} \sum^D_{i,j = 1} \log\bigg(\frac{x_i}{x_j}\bigg)\log\bigg(\frac{y_i}{y_j}\bigg) w_{ij}
\end{align*}
These weights can be fixed before the analysis to only let the information of certain important log-ratios play a role, or they can be chosen in a data dependent way. 

The idea of weighting in CoDa has been considered before, see \citet{van2014bayes}, \citet{hron2017weighted}, \citet{greenacre2019variable}, \citet{greenacre2021comparison} and \citet{hron2021weighted}. \citet{greenacre2019variable} also considers only keeping a few log-ratios which represent the data well and draws connections to graph theory.
Our contribution differs, however, from the mentioned ones in that the underlying geometry is adapted through distinct $w_{ij}$, and that we are able after transforming the data to work in the standard Euclidean geometry. 

Before introducing the new concepts in detail, we will recapture some important results of graph theory.

\subsection{Some results from graph theory}

In this subsection we use the notation $(f_1,\ldots,f_D)\in \mathbb{R}^D$ and $(g_1,\ldots,g_D)\in \mathbb{R}^D$ for two sets of variables, in order 
to avoid any confusion with the compositional case. We define a graph as a fixed pair $(\mathcal{V},\matr{W})$, where $\mathcal{V}:=\{1,\ldots,D\}$ denotes a set of indices, and $\matr{W} = (w_{ij})_{\substack{1 \leq i,j \leq D }} \in \mathbb{R}^{D\times D}$ is a symmetric matrix with zero diagonal and non-negative entries, corresponding to weights between indices. 
Graphs are useful to model the relation between variables $(f_1,\ldots,f_D)\in \mathbb{R}^D$. 
The idea is that the bigger a weight $w_{ij}$, the bigger the relationship between the two variables $f_i$ and $f_j$ is. Whenever $w_{ij}$ is zero there is no relation. 

The edge-set of a graph is defined as $\mathcal{E}:= \{ (i,j) \mid w_{ij} \neq 0 \} \subset \mathcal{V}^2$. We write in the following  $i \sim j$ whenever $(i,j) \in \mathcal{E}$. We say that there exists a path from the vertex $i$ to the vertex $j$ if there are vertices $i_1,i_2,\ldots,i_k$, with $i = i_1, j=i_k $, and $w_{i_1 i_2} \neq 0 ,w_{i_2 i_3} \neq 0 ,\ldots,w_{i_{k-1} i_k}\neq 0$. A subset of indices $\{i_1,\ldots,i_k\} \subset \mathcal{V} $ is called connected if there is a path from each vertex in the subset to another vertex in the subset. If the subset is equal to $\mathcal{V}$ then the graph is said to be connected, otherwise we say it is disconnected. The set of vertices $\mathcal{V}$ can always be written as a union of its connected components $\mathcal{V} =  \cup_{m=1}^M \mathcal{V}_m$, with disjoint sets $\mathcal{V}_1,\ldots,\mathcal{V}_M$, for any graph. Such a decomposition can be obtained by starting with one vertex, say $v_1 = 1$, and looking for all other vertices connected with $v_1$ through a path to obtain $\mathcal{V}_1$. Deleting $\mathcal{V}_1$ from $\mathcal{V}$ we can restart the procedure to get $\mathcal{V}_2$, and so on.  

An important analytical tool for finite graphs is the so called Laplacian-matrix, defined as
\begin{align} \label{mainform}
\matr{L}_{\matr{W}}:= 
\diag{\bigg(\sum_{1 \sim j} w_{1j},\ldots ,\sum_{D\sim j} w_{Dj}\bigg)} - \matr{W}
\end{align}
where $\diag$ is the diagonal matrix of the corresponding entries.
The definition of $\matr{L}_{\matr{W}}$ is motivated by the following key equality
\begin{align} \label{keyEq}
    \frac{1}{2}\sum_{(i,j) \in \mathcal{E}} (f_i-f_j)(g_i-g_j)w_{ij} = \vec{f}' \matr{L}_{\matr{W}} \vec{g},
\end{align}
for any $\vec{f}, \vec{g} \in \mathbb{R}^D$, see \citet{merris1994laplacian}.


\begin{remark}
We do not necessarily need to assume that $w_{ij} = w_{ji}$ holds. Assume $u_{ij}$ are non-negative weights, with $u_{ii}=0$, but not necessarily  symmetric. Then we can define symmetric weights  by setting $w_{ij}:=\frac{1}{2}(u_{ij}+u_{ji})$.
For these $w_{ij}$ we have
\begin{align*}
\frac{1}{2}\sum_{i,j = 1}^D (f_i-f_j)(g_i-g_j)u_{ij} &= \frac{1}{2}\Bigg(\sum_{i<j} (f_i-f_j)(g_i-g_j)u_{ij} + \sum_{j<i} (f_i-f_j)(g_i-g_j)u_{ij}\Bigg)\\  &= \frac{1}{2} \sum_{i<j} (f_i-f_j)(g_i-g_j)(u_{ij}+u_{ji})  \\
&= \frac{1}{2} \sum_{i,j=1}^D (f_i-f_j)(g_i-g_j)w_{ij}.
\end{align*}
\end{remark}

Another important tool in graph theory is the incidence matrix $\matr{d}_{\matr{W}} \in \mathbb{R}^{\mid\mathcal{E}\mid \times \mid V \mid}$. 
For fixed weights $\matr{W}$ it is defined as 
\begin{align*}
    (\matr{d}_{\matr{W}} )_{e,l}:=  
    \begin{cases} 
      w_{lj} & e = (l,j) \\
      -w_{il} & e = (i,l) \\
      0 & \text{else} .
   \end{cases}
\end{align*}
The incidence matrix defines a graph theoretic  analog to usual differentiation along weighted edges. For a fixed edge $e = (l,j)$ we get $( \matr{d}_{\matr{W}} \vec{f})_e = (w_{lj}(f_l-f_j))$ which measures the weighted difference between values in the nodes $l$ and $j$.
A standard result in graph theory is the equality $\matr{d}_{\matr{W}^{\frac{1}{2}}}' \matr{d}_{\matr{W}^{\frac{1}{2}}} = 2\matr{L}_{\matr{W}}$, where the superscript $\frac{1}{2}$ is understood as a coordinate-wise operation. With the latter property we can write Equation (\ref{keyEq}) also as
\begin{align} \label{keqEqExtended}
    \frac{1}{2}\sum_{(i,j) \in \mathcal{E}} (f_i-f_j)(g_i-g_j)w_{ij} = \vec{f}' \matr{L}_{\matr{W}} \vec{g} = \frac{1}{2} \langle \matr{d}_{\matr{W}^{\frac{1}{2}}} \vec{f} , \matr{d}_{\matr{W}^{\frac{1}{2}}} \vec{g} \rangle_{2}.
\end{align}

For a more thorough introduction to graph theory we refer to \citet{gross2005graph} or \citet{chung1997spectral}.

The properties of $\matr{L}_{\matr{W}}$ have been analyzed extensively. 
In this paper we will need the following standard results in spectral graph theory, see \citet{Mohar91thelaplacian} for a proof.

\begin{lemma} \label{LaplacianLemma}
Assume that $\matr{W}$ is symmetric with zero diagonal and non-negative entries, then:
\begin{itemize}
   \item  $\matr{L}_{\matr{W}}$ is a symmetric positive semi-definite matrix.
    \item The vector of all ones $\matr{1} := (1,\ldots,1)' \in \mathbb{R}^D$  is always an eigenvector to zero of $\matr{L}_{\matr{W}}$, i.e.~$\matr{L}_{\matr{W}} \matr{1} = \vec{0}$.
    \item If $\mathcal{V}$ is the union of more than one connected component, $\mathcal{V} =  \cup_{m=1}^M \mathcal{V}_m$, then for each connected component $\mathcal{V}_m$, the vector in $\mathbb{R}^D$ with ones at position $i \in \mathcal{V}_m$ and otherwise zeros, $\matr{1}_{i \in \mathcal{V}_m} \in \mathbb{R}^D$, is an eigenvector to zero, i.e. $\matr{L}_{\matr{W}}\matr{1}_{i \in \mathcal{V}_m} = \vec{0}$.
    The vectors $\vec{1}_{i \in \mathcal{V}_m}$ span the kernel of $\matr{L}_{\matr{W}}$.
    \item There exists a permutation matrix $\matr{P}$ such that $\matr{P}\matr{L}_{\matr{W}}\matr{P}'$ is in block diagonal form with blocks $\matr{L}_1,\ldots,\matr{L}_M$ for the weights  $\matr{P}\matr{W}\matr{P}'$.
    \item The Laplacian-matrix $\matr{L}_{\matr{W}}$ has exactly $M$ many zero eigenvalues.
\end{itemize}
\end{lemma}

From now on we will assume for non-connected graphs that the Laplacian-matrix $\matr{L}_{\matr{W}}$ is always in block diagonal form. This can always be achieved by simply relabeling the vertices using Lemma~\ref{LaplacianLemma}.

\section{Compositional data on graphs}

\subsection{Graph simplex space, norms and inner products}

Equation (\ref{keyEq}) allows us to make the connection to compositional data: When taking the weight matrix
 $\matr{W}= \frac{1}{D}(\vec{1} \vec{1}' - \matr{I}_D) $, $w_{ij} = \frac{1}{D}$, where $\matr{I}_D$ denotes the identity matrix of dimension $D$, and  $\matr{L}_{A}:= (1-\frac{1}{D})\matr{I}_D - \frac{1}{D}(\vec{1} \vec{1}' - \matr{I}_D) = \matr{I}_D - \frac{1}{D}\vec{1} \vec{1}'$, we recover for any $\vec{x},\vec{y} \in \mathcal{S}^D$
\begin{align*}
    \langle \vec{x},\vec{y} \rangle_{\mathcal{A}} &= \log(\vec{x}) ' \matr{L}_{\mathcal{A}} \log(\vec{y}).
\end{align*}
$\matr{L}_{A}$ is known as the centering matrix, see \citet{marden2019analyzing}. It has exactly one eigenvalue equal to zero, with eigenvector $\vec{1}$. All other eigenvalues are equal to 1. The eigenvector $\vec{1}$ corresponds to the null space of $\matr{L}_{A}$ and so the rescaling invariance of the Aitchison inner product is a direct consequence of the Laplacian matrix $\matr{L}_{A}$, as  $\matr{L}_{A}\log(c\odot \vec{ x}) = \matr{L}_{A}(\log(\vec{x}) + \log(c) \vec{1}) = \matr{L}_{A}\log(\vec{ x})$ holds for any positive constant $c$. 
We can see that on  $\mathbb{R}^D_+$ the bilinear form $\log(\vec{x}) ' \matr{L}_{\mathcal{A}} \log(\vec{y})$ is not an inner product as we can only deduce from $\log(\vec{x})' \matr{L}_{\mathcal{A}} \log(\vec{x}) = 0$ that $\log(\vec{x})$ is in the null space of $\matr{L}_{A}$. Instead of considering quotient spaces and modified operations, an additional condition, such as $\sum_{j=1}^D x_j = 1$, is used in compositional data.

Given a graph $(\mathcal{V},\matr{W})$, with a partition into connected components $\mathcal{V} = \cup_{m = 1}^M \mathcal{V}_m$, this leads to the definition of the \textit{D-part Graph Simplex} as
\begin{align}
    \mathcal{S}^D_{\matr{W}}:=\bigg\{ (x_1,\ldots,x_D)' \in \mathbb{R}^D_{+} \bigm\vert \sum_{j \in \mathcal{V}_m} x_j = \kappa_m,  m = 1,\ldots,M  \bigg\}, \label{graphSimplex}
\end{align}
for some $\kappa_1,\ldots,\kappa_M > 0$, and scaled versions of the perturbation and the powering operations such that the latter map into $\mathcal{S}^D_{\matr{W}}$:

\begin{itemize}
    \item 
    $(\vec{x} \oplus_{\matr{W}} \vec{y})_{i \in \mathcal{V}_m} := \frac{\kappa_m}{\sum_{j \in \mathcal{V}_m} x_j y_j} (x_i y_i)_{i \in \mathcal{V}_m}'  $
    \item $(\alpha \odot_{\matr{W}} \vec{x})_{i \in \mathcal{V}_m} := \frac{\kappa_m}{\sum_{_{j \in \mathcal{V}_m}} x_j^\alpha} (x_i^\alpha)_{i \in \mathcal{V}_m }' $
\end{itemize}
for all $m = 1,\ldots,M$, where the subscript $i \in \mathcal{V}_m$ denotes the entries with index in $ \mathcal{V}_m$. Note that in the definition of $\mathcal{S}^D_{\matr{W}}$ other conditions could be used, provided that the perturbation and the power operation is changed accordingly.
The natural extension of the Aitchison inner product to a graph structure on $\log(\vec{x})$ is to equip $\mathcal{S}^D_{\matr{W}}$ with the inner product $\log(\vec{x})' \matr{L}_{\matr{W}} \log(\vec{y})$.
In more generality we define for a fixed $\alpha \geq 0 $ with non-negative entries
\begin{align} 
    &\langle \vec{x}, \vec{y} \rangle_{\matr{W},\alpha}:=
    \alpha \langle \log(\vec{x}), \log(\vec{y})  \rangle_{2} + \langle  \log(\vec{x}),\matr{L}_{\matr{W}} \log(\vec{y})  \rangle_{2} \label{fullInner} \\
    & \norm{\vec{x}}_{\matr{W},\alpha}:= \sqrt{\langle \vec{x}, \vec{x} \rangle_{\matr{W},\alpha}}
\end{align}
for any $\vec{x},\vec{y} \in \mathbb{R}^D_+$. 

This definition is motivated by Equation (\ref{keqEqExtended}), $\langle  \vec{f},\matr{L}_{\matr{W}} \vec{g}  \rangle_{2} =  \frac{1}{2} \langle \matr{d}_{\matr{W}^{\frac{1}{2}}} \vec{f} , \matr{d}_{\matr{W}^{\frac{1}{2}}} \vec{g} \rangle_{2}$. The incidence matrix $\matr{d}_{\matr{W}^{\frac{1}{2}}}$ can be seen as the graph analogue to directional differentiation of real valued functions, \citet{ostrovskii2005sobolev}, see \citet{grady2010discrete} for an introduction to calculus on graphs. Therefore, $\langle  \vec{f},\matr{L}_{\matr{W}} \vec{g}  \rangle_{2}$ can be thought of as the graph analogue to the bilinear form $ (f,g) \mapsto \int f'(x)g'(x) dx$, for functions $f,g :\mathbb{R} \rightarrow \mathbb{R}$ in a suitable function space, giving rise to the Sobolev semi-norm in functional analysis, see \citet{adams2003sobolev}. Adding $\int f(x)g(x)dx$ to $\int f'(x)g'(x)dx$ turns it into an inner product. 

\begin{lemma} \label{graphCodaInnerProd}
The space $(\mathbb{R}^D_+,\oplus,\odot)$  equipped with $\langle \vec{x}, \vec{y} \rangle_{\matr{W},\alpha}$, for fixed $\alpha > 0$ is a Hilbert space. 
Similarly, $(\mathcal{S}^D_{\matr{W}},\oplus_{\matr{W}},\odot_{\matr{W}})$ equipped with $\langle \vec{x}, \vec{y} \rangle_{\matr{W},\alpha}$ for $\alpha = 0$ is a Hilbert space. 
For $\alpha = 0$ we recover the inner product 
$  \frac{1}{2}\sum_{(i,j)\in \mathcal{E}} \log\big(\frac{x_i}{x_j}\big)\log \big(\frac{y_i}{y_j}\big) w_{ij}$.

\end{lemma}

The proof can be found in the Appendix. 

\begin{remark}
Extensions of $\langle \vec{x}, \vec{y} \rangle_{\matr{W},\alpha}$ are possible by adding $ \langle  \log(\vec{x}),\matr{L}^m_{\matr{W}} \log(\vec{y})  \rangle_{2} $ to the former for some $m \in  \mathbb{N}, m >1 $. In the case of $m = 2$ this gives the additional term $\sum_{i = 1}^D \log\bigg(\frac{x_i^{d(i)}}{\prod_{i\sim j } x_j^{w_{ij}  }}\bigg)    \log\bigg(\frac{y_i^{d(i)}}{\prod_{i\sim j } y_j^{w_{ij}  }}\bigg)$, where $d(i):=\sum_{i\sim j} w_{ij}$. Again one can show that $(\mathcal{S}^D_{\matr{W}},\oplus_{\matr{W}},\odot_{\matr{W}})$ equipped with such an inner product is a Hilbert space. 
\end{remark}

Equation (\ref{keqEqExtended}) also motivates to extend the inner product $ \langle \vec{x}, \vec{y} \rangle_{\matr{W},\alpha}$ to $q$-norms. Define 
the standard $q$-norm on the Euclidean space -- $\vec{f} \in \mathbb{R}^D$:
\begin{align*}
    \norm{\vec{f}}_{q}:= 
    \begin{cases} 
     \big( \sum_{i = 1}^D \mid f_i \mid^q \big)^\frac{1}{q} &  \text{for } 1\leq q < \infty \\
      \max_{i = 1,\ldots,D} \mid f_i \mid & \text{for } q = \infty \\
   \end{cases}.
\end{align*}
Then for any $\vec{x} \in \mathbb{R}^D_+$ we define
\begin{align*}
    \norm{ \vec{x} }_{q,\alpha} :=
    \begin{cases}
       \bigg( \alpha \norm{\log(\vec{x})}^q_q + \frac{1}{2}\norm{\matr{d}_{\matr{W}^{\frac{1}{q}}} \log(\vec{x})}^q_q \bigg)^\frac{1}{q}  \text{for } 1\leq q < \infty \\
     \max{\bigg(\alpha \norm{\log(\vec{x})}_q,\frac{1}{2}\norm{\matr{d}_{\matr{W}} \log(\vec{x})}_q,\bigg)} \text{for }q =  \infty  &
    \end{cases}.
\end{align*}
For fixed $\alpha$ with $\alpha > 0$ resp. $\alpha = 0$, $\norm{ \cdot }_{q,\alpha}$ is a norm on $(\mathbb{R}^D_+,\oplus,\odot)$ resp. $(\mathcal{S}^D_{\matr{W}},\oplus_{\matr{W}},\odot_{\matr{W}})$. The proof is in the Appendix.

\subsection{ Centered and isometric log-ratio transforms  } \label{clrAndIlr}

In the classical compositional setting, where $\matr{L}_{\matr{W}} = \matr{L}_{\mathcal{A}}$, the $\clr$ and the $\ilr_{\matr{V}}$ mappings play an important role in establishing isometry of $\mathcal{S}^D$ to $\mathbb{R}^{D-1}$,
and they are often used to interpret results. 

In the notation of the Laplacian-matrix the $\clr$ map is given by
\begin{align*}
   \clr{(\vec{x})}  = \matr{L}_{\mathcal{A}}  \log(\vec{x}) . 
\end{align*}
This motivates us to define the weighted clr map as
\begin{align*}
     \clr{(\vec{x})}_{\matr{W},\alpha}:=\big(\alpha \matr{I} +\matr{L}_{\matr{W}}\big)^{\frac{1}{2}} \log(\vec{x}).
\end{align*}
In the compositional setting we have $\matr{L}_{\mathcal{A}} = \matr{L}_{\mathcal{A}}^2$, and from this follows (\ref{innerIso}).
Similarly, the mapping $\vec{x} \mapsto \clr{(\vec{x})}_{\matr{W},\alpha}$ is also distance preserving as $\langle \vec{x}, \vec{y} \rangle_{\matr{W},\alpha} =  \langle  \big(\alpha \matr{I} +\matr{L}_{\matr{W}}\big)^{\frac{1}{2}} \log(\vec{x}),\big(\alpha \matr{I} +\matr{L}_{\matr{W}}\big)^{\frac{1}{2}} \log(\vec{y})  \rangle_{2} $ holds for any $\vec{x},\vec{y}$.
Note that for $\langle\cdot, \cdot \rangle_{\matr{W},\alpha}$
the matrix  $\alpha \matr{I} +\matr{L}_{\matr{W}}$ is not distance preserving. However its interpretation might be easier as $((\alpha \matr{I} +\matr{L}_{\matr{W}})\log(\vec{x}))_j$ is equal to $ \log\Big(\frac{x_j^{\alpha +  d_j}}{\prod_{i \sim j} x_i^{ w_{ij}} }\Big)$, with $d_j:= \sum_{i \sim j} w_{ij}$, for any $\vec{x} \in \mathbb{R}^D_+$.
The interpretation of the latter is that each variable is centered by its weighted neighborhood. The square root, although distance preserving is not necessarily 
 sparse.
Instead, we construct more interpretable one-to-one mappings similar to the compositional case.
For specific choices of $\matr{V}$  the ilr map (\ref{ilrCoda}) leads to a very interpretable one-to-one isometry, see \citet{filzmoser2018applied} or \citet{fivserova2011interpretation}, 
\begin{align} \label{ilrpivot}
    (\ilr(\vec{x})_{\matr{V}})_j:=\sqrt{\frac{D-j}{D-j+1}} \log\Big( \frac{x_j}{\sqrt[D-j]{\prod_{i=j+1}^D x_i}} \Big).
\end{align}
The $j$-th coordinate $ (\ilr(\vec{x})_{\matr{V}})_j$ only incorporates information of $x_i$ with $i \geq j$. Looking at the coordinate for $j=1$ leads to a high interpretability as $x_i$ only appears in the first coordinate.

In the graphical setting we take the following approach  to obtain an interpretable isometry.
The idea is to use a modified version of the Cholesky decomposition for positive semi-definite matrices to write $\matr{L}_{\mathcal{A}} = \matr{C}'\matr{C}$, where $\matr{C}$ is an upper triangular matrix with non-negative diagonal elements. The $j$-th entry of the mapping $\vec{x} \mapsto \matr{C}\log(\vec{x})$ contains only the information of $\log(x_{j}),\ldots,\log(x_{D})$ as a weighted sum similar to (\ref{ilrpivot}).

\begin{lemma} \label{cholForgraph}
%
Denote $\matr{L}_1,\ldots,\matr{L}_M$ the diagonal blocks of $\matr{L}_{\matr{W}}$. Then there exist upper triangular matrices $\matr{C}_1,\ldots,\matr{C}_M$ 
having as last row entirely zeros and all diagonal elements, except the last one, positive
such that
\begin{align*}
   \matr{L}_{\matr{W}} =\matr{C}'\matr{C} \ , 
\end{align*} where $\matr{C} = 
\diag{(\matr{C}_1,\ldots,\matr{C}_M)}$.
The matrices  $\matr{C}_m$, for $m=1\ldots ,M$
fulfill $\matr{C}_m \vec{1} = \vec{0}$. Each matrix $\matr{C}_m$ can be obtained by computing the eigen-decompositions of each $\matr{L}_m = \matr{U}_m \matr{\Sigma}_m  \matr{U}_m'$, continued by the QR decompositions  $ \matr{U}_m\matr{\Sigma}^{\frac{1}{2}}_m \matr{U}_m' = \matr{Q}_m \matr{R}_m$, and setting, if necessary after a sign change of the diagonal elements, $\matr{C}_m = \matr{R}_m$.

\end{lemma}

Lemma \ref{cholForgraph} gives a mapping which puts the emphasis on the first coordinate of each subgraph. By permuting the coordinates of each subgraph via permutation matrices we can put the focus on the coordinates of interest. All together we get the following.

\begin{lemma}[Graph Isometric Log-ratio map (GILR1)] \label{graphILR}

Denote  $\matr{P}_1,\ldots,\matr{P}_M$ permutations of $\mathcal{V}_1,\ldots,\mathcal{V}_M$, and $\matr{C}_1,\ldots,\matr{C}_m$ the Cholesky decomposition  of the Laplacian matrices $\matr{P}_m\matr{L}_{m}\matr{P}_m ' $ as in Lemma \ref{cholForgraph}. Then for $\alpha = 0$, the map 
\begin{align} \label{ilrReduced}
    \vec{x} \mapsto \diag{(\matr{C}_{-\mid \mathcal{V}_1 \mid },\ldots,\matr{C}_{-\mid \mathcal{V}_M \mid })} \diag{(\matr{P}_1,\ldots,\matr{P}_M)}  \log(\vec{x}),
\end{align}
where $\matr{C}_{-\mid \mathcal{V}_m \mid}$ denotes the matrix $\matr{C}_m$ after deletion of the last row,
is a linear one-to-one isometry from $(\mathcal{S}^D_{\matr{W}},\oplus_{\matr{W}},\odot_{\matr{W}})$, equipped with $\langle \cdot, \cdot \rangle_{\matr{W},0}$, to $(\mathbb{R}^{D-M},+,\cdot)$ equipped with $\langle \cdot,\cdot \rangle_{2}$.
For $\alpha >0 $ compute the Cholesky decompositions  $\alpha\matr{I} +  \matr{P}_m\matr{L}_m\matr{P}_m' =  \Tilde{\matr{C}}_m'\Tilde{\matr{C}}_m$, where $\Tilde{\matr{C}}_m$ are upper triangular matrices with strictly positive diagonals, then
\begin{align} \label{ilrAll}
    \vec{x} \mapsto  \diag{(\Tilde{\matr{C}}_1,\ldots,\Tilde{\matr{C}}_M)} \diag{(\matr{P}_1,\ldots,\matr{P}_M)} \log(\vec{x})
\end{align}
is a linear one-to-one isometry from $(\mathbb{R}^D_+,\oplus,\odot)$, equipped with $\langle \cdot, \cdot \rangle_{\matr{W},\alpha}$, to $(\mathbb{R}^{D},\langle \cdot,\cdot \rangle_{2},+,\cdot)$, equipped with $\langle \cdot,\cdot \rangle_{2}$.
\end{lemma}

To obtain an orthogonal system $\log(\vec{v}_{i}) \in \mathcal{S}^D_{\matr{W}}$ resp. $ \log(\vec{v}_{i}) \in \mathbb{R}^D_+$, with $i = 1,...,D-M$, where $M$ is zero in the second case, for the inner product $\langle \vec{x}, \vec{y} \rangle_{\matr{W},\alpha}$ we can solve the linear equations $F(\vec{v}_{i}) = \vec{1}_{j \in \{i\}}$ where $F$ is the map introduced in (\ref{ilrReduced}) resp. (\ref{ilrAll}). Then for any $k \neq l$ it follows $\langle \vec{v}_{k}, \vec{v}_{l} \rangle_{\matr{W},\alpha} = \langle F(\vec{v}_{k}),F(\vec{v}_{l})  \rangle_2 = \langle  \vec{1}_{j \in \{k\}},  \vec{1}_{j \in \{l\}}  \rangle_2 = \delta_{kl}$.

The matrices in (\ref{ilrAll}) and (\ref{ilrReduced})  of Lemma \ref{graphILR} do not have orthogonal rows, however, they are just as interpretable as in the usual compositional data case. For simplicity assume that the graph is connected. After the choice of a permutation matrix $\matr{P}$, which can be identified with the mapping $\pi:\{1,\ldots,D\} \rightarrow \{1,\ldots,D\}$, we see that for (\ref{ilrReduced}), by $\matr{C}_1 \vec{1} = 0$, the first element is $ \log \Big(\frac{x_{\pi(1)}^{c_{11}}}{\prod_{i=2}^D x_{\pi(i)}^{-c_{1i}}} \Big) =  \log \Big(\prod_{i=2}^D \big(\frac{x_{\pi(1)}}{ x_{\pi(i)}}\big)^{-c_{1i}} \Big)$, the second element $ \log \Big(\frac{x_{\pi(2)}^{c_{22}}}{\prod_{i=3}^D x_{\pi(i)}^{-c_{2i}}} \Big) =  \log \Big(\prod_{i=3}^D \big(\frac{x_{\pi(2)}}{ x_{\pi(i)}}\big)^{-c_{1i}} \Big)$ and so on, up to $ \log \Big(\frac{x_{\pi(D-1)}^{c_{(D-1)(D-1)}}}{ x_{\pi(D)}^{-c_{(D-1)D}}} \Big) =   \log \Big(\big(\frac{x_{\pi(D-1)}}{ x_{\pi(D)}}\big)^{-c_{(D-1)D}} \Big)$.

For (\ref{ilrAll}) we get the same expressions as above on the left side of the equalities, and in addition the last element is $c_{DD} \log(x_{\pi(D)})$.
Other one-to-one isometries can be constructed.

\begin{lemma}[Graph Isometric Log-ratio map 2 (GILR2)]
Taking an eigen-decomposition of $\alpha \matr{I} +\matr{L}_{\matr{W}}$, $\alpha \matr{I} +\matr{L}_{\matr{W}} = \matr{U} \matr{\Sigma} \matr{U}' $, where the diagonal elements $\lambda_1,\ldots ,\lambda_D$ of $\matr{\Sigma}$ are ordered from biggest to smallest and denoting the i-th column of $\matr{U}$ by $\vec{u}_i$ and $M$ the number of zero eigenvalues, we can define a one-to-one isometric map, $(\gilr_{\matr{W}})$, by  
\begin{align} \label{ilr2}
    \vec{x} \mapsto ( \sqrt{\lambda_i} \langle  \vec{u}_i,\log(\vec{x}) \rangle_2 )_{i \in 1,\ldots ,D-M}' \in \mathbb{R}^{D-M}.
\end{align}
 Its inverse is given by $\vec{z} \mapsto \sum_{i = 1}^{D-M} z_i \frac{1}{\sqrt{\lambda_i}} \vec{u}_i  $.
\end{lemma}

\begin{proof}
The proof of this Lemma follows directly from the eigen-decomposition. The mapping (\ref{ilr2}) is per definition an isometry and therefore also injective. Surjectivity follows as $\vec{u}_i$ form an orthonormal system. The inverse can easily be checked by plugging in one expression into the other; see also the next section.
\end{proof}

\begin{remark}[Fourier transform]
The mapping from the previous Lemma has an interesting interpretation for $\alpha = 0$. The projections $ (\langle \vec{u}_i ,\log(\vec{x}) \rangle )_{i \in 1,\ldots ,D-k}'$ onto the i-th eigenvector of $\matr{L}_{\matr{W}}$ can be interpreted as a discrete analogue of the Fourier transform on graphs to the frequencies $\lambda_i$, see \citet{shuman2016vertex}. The projections corresponding to small $\lambda_i$ coincide with the smooth part of a signal on a graph, whereas the projections corresponding to high $\lambda_i$ coincide with the higher frequency part of the signal $\log(\vec{x})$. 
The eigenvectors are the non-trivial minimizers of $\sum_{i \neq j} (\log(x_i) - \log(x_j))^2w_{ij}$, such that $\norm{\log(\vec{x})}_2 = 1$.
For compositional data all non-zero frequencies are one and therefore a distinction is not useful. 
\end{remark}

\begin{example}[The star graph]
Assume that, w.l.o.g., we are interested in modeling only the dependence between the variable $x_1$ linked to all others. We fix the weights $w_{1j} = w_{ji} = 1$ for all $j \neq 1$  and zero otherwise. This corresponds to putting weights 1 onto the log-ratios $\log(\frac{x_1}{x_2}),\log(\frac{x_1}{x_3}),\ldots ,\log(\frac{x_1}{x_D})$ in (\ref{fullInner}). The corresponding Laplacian-matrix $\matr{L}_{\matr{W}}$ is of the form
\begin{align*}
    \matr{L}_{\matr{W}} = 
    \begin{pmatrix}
D-1 &  -\vec{1}' \\
-\vec{1} &  \matr{I}_{D-1} ,
\end{pmatrix}
\end{align*}
where $\matr{I}_{D-1}$ is the D-1 dimensional unity matrix. This matrix has one zero eigenvalue, one eigenvalue equal to $D$, and $D-2$ eigenvalues equal to one. The (non-normalized) eigenvector to $D$ is $(D-1,-1,\ldots ,-1)'$ and one (non-orthogonal) system of eigenvectors to one is given by $(0,1,-1,0, \ldots ,0)',\\
(0,0,1,-1,\ldots ,0)',\ldots ,(0,\ldots ,1,-1)$; see \citet{grone1990laplacian}.

\end{example}

\section{A further extension} \label{extension}

We can generalize the theory developed for the spaces  $(\mathcal{S}^D_{\matr{W}},\langle \cdot, \cdot \rangle_{\matr{W},\alpha},\oplus_{\matr{W}},\odot_{\matr{W}})$ even further. One important extension that comes to mind is to allow the weights $w_{ij}$ to also take negative values, which can be seen as an extension for the following two important examples:

\begin{example} \label{ex2}
Assume that through expert knowledge we are interested in analyzing only certain weighted combinations, say $\sum_{j = 1}^D \log(\frac{x_i}{x_j}) w_{ij}$, for $i \in \{i_1,...,i_L \} \subset \{1,...,D\}$ and $w_{ij} \in \mathbb{R}$, not necessarily symmetric. Define a, rectangular, matrix $\widetilde{\matr{L}}_{\matr{W}} \in \mathbb{R}^{L \times D}$ by $(\widetilde{\matr{L}}_{\matr{W}})_{ii} = \sum_{j = 1}^D w_{ij}$ for  $i \in \{i_1,...,i_L \}$ and $(\widetilde{\matr{L}}_{\matr{W}})_{ij} = -w_{ij}$ , for $i \neq j$. The matrix $\matr{L}_{\matr{W}}:= \widetilde{\matr{L}}_{\matr{W}} (\widetilde{\matr{L}}_{\matr{W}})' $ is symmetric, positive semi-definite, has real valued entries and $\vec{1}$ in its null space. The natural mapping $\vec{x} \mapsto \widetilde{\matr{L}}_{\matr{W}} \log(\vec{x})$ can be used as a building block as explained below.
\end{example}

\begin{example} \label{ex3}
Assume that we are interested in analyzing combinations of certain subsets of ratios, say $\sum_{i,j = 1}^D \log(\frac{x_i}{x_j})w_{ijk}$ for  $k=1,...,K$. Then for each $k$ we can write the latter as
$\log(\vec{x})'\vec{w}_k$ with a vector $\vec{w}_k$ defined by $(\vec{w}_k)_i = \sum_{j=1}^D w_{ijk} - \sum_{j=1}^D w_{jik} $, for $i =1,...,D$. By definition of the combinations we directly have $\vec{w}_k' \vec{1} = 0$ for each $k$. If we define the rows of a rectangular matrix  $\widetilde{\matr{L}}_{\matr{W}} \in \mathbb{R}^{K \times D}$ as $\vec{w}_k$ we get $\widetilde{\matr{L}}_{\matr{W}} \log(\vec{x})  = (\sum_{i,j = 1}^D \log(\frac{x_i}{x_j})w_{ijk})_{k =1,..,K}$. Again, as for the previous example, the matrix $\matr{L}_{\matr{W}}:= \widetilde{\matr{L}}_{\matr{W}} (\widetilde{\matr{L}}_{\matr{W}})' $ is symmetric, positive semi-definite, has real valued entries and $\vec{1}$ in its null space and the natural mapping $\vec{x} \mapsto \widetilde{\matr{L}}_{\matr{W}} \log(\vec{x})$ can be used as a building block of an inner product as explained next.
\end{example}

We drop the assumption of non-negative weights $w_{ij}$ but assume $\matr{L}_{\matr{W}}$ to be positive semi-definite. The goal is to again be able to define scale-invariant spaces. We start by defining the equivalence relation $\log(\vec{x}) \sim \log(\vec{y}) \Leftrightarrow \log(\vec{x}) - \log(\vec{y}) \in \text{Ker}(\matr{L}_{\matr{W}})$, where $\text{Ker}$ denotes the kernel of the linear operator $\matr{L}_{\matr{W}}$. This equivalence relation induces the equivalence classes $[\log(\vec{x})]:=\{\log(\vec{y}) \mid \log(\vec{x}) \sim \log(\vec{y}) \}$. In the following we write $"\equiv"$ whenever two elements belong to the same equivalence class. From the theory of quotient spaces, see \citet{roman2005advanced}, we get that the space of equivalence classes, namely $\{[\log(\vec{x})] \mid \vec{x} \in  \mathbb{R}^{D}_+  \}$, which can be also written as $[\log(\vec{x})] = \{ \log(\vec{x}) + \vec{q} \mid \vec{q} \in \text{Ker}(\matr{L}_{\matr{W}}) \}$, is a linear vector space equipped with the operations $[\log(\vec{x})] + [\log(\vec{y})]:= [\log(\vec{x}) +\log(\vec{y}) ]$ and $\alpha [\log(\vec{x})] := [\alpha \log(\vec{x})]$ for any $\vec{x}, \vec{y} \in \mathbb{R}^D_+$ and $\alpha \in \mathbb{R}$. Therefore to define a scale invariant Hilbert space on $\mathcal{S}^D_{\sim}:= \{ \exp([\log(\vec{x})]) \mid \vec{x} \in \mathbb{R}^D_+ \}$ we set
\begin{align}
&\vec{x} \oplus_{\sim} \vec{y}:= \exp([\log(\vec{x})] + [\log(\vec{y})]) \label{equivplus} \\    
&  \alpha \odot_{\sim} \vec{x} := \exp(\alpha [\log(\vec{x})]) \label{equivtimes} \\
& \langle \vec{x}, \vec{y} \rangle_{\sim} := \langle [\log(\vec{x})], \matr{L}_{\matr{W}}[\log(\vec{y})] \rangle_2 \label{equivIP}
\end{align}
for $\vec{x},\vec{y} \in \mathcal{S}^D_{\sim}$ and $\alpha \in \mathbb{R}$. 


To prove that $\mathcal{S}^D_{\sim}$ is a Hilbert space we only need to check that (\ref{equivIP}) is indeed an inner product. Linearity follows trivially from the definitions (\ref{equivplus}) and (\ref{equivtimes}) as the log is taken in (\ref{equivIP}). The property $ \langle \vec{x}, \vec{x} \rangle_{\sim} \geq 0$ holds for any $\vec{x} \in \mathcal{S}^D_{\sim}$ as $\matr{L}_{\matr{W}}$ is assumed to be positive semi-definite. When $\langle \vec{x}, \vec{x} \rangle_{\sim}  = 0$ we get by taking the eigen-decomposition of $\matr{L}_{\matr{W}} = \matr{U}\matr{\Sigma}\matr{U}'$ and writing (\ref{equivIP}) as $\sum_{i = 1}^{D-M} (\langle \vec{u}_i,[\log(\vec{x})] \rangle_2 )^2$ that $[\log(\vec{x})]$ is orthogonal to the eigenvectors of the non-zero eigenvalues of $\vec{u}_i$, and therefore 
$[\log(\vec{x})] \in \text{Ker}(\matr{L}_{\matr{W}})$, $\log(\vec{x}) \in [\vec{0}]$. Finally note that scale invariance on any subgraph is still given by the definition of $\matr{L}_{\matr{W}}$.
For any $\vec{x} \in \mathbb{R}^D_+$ we can write $\log(\vec{x})= \sum_{i=1}^{D-M} \langle \vec{u}_i,\log(\vec{x}) \rangle \vec{u}_i + \vec{q}$, with $\vec{q}\in \text{Ker}(\matr{L}_{\matr{W}})$, and so $\log(\vec{x}) \in [\sum_{i=1}^{D-M} \langle \vec{u}_i,\log(\vec{x}) \rangle \vec{u}_i ]$.
To obtain an isometric one-to-one map as in (\ref{ilr2}) we can again define 
$\vec{x} \mapsto (\sqrt{\lambda_i} \langle \vec{u}_i,\log(\vec{x}) \rangle)_{i=1,...,D-M} \in \mathbb{R}^{D-M}$ for $ \vec{x} \in \mathcal{S}^D_{\sim}$ where $M$ is the number of zero eigenvalues and $\lambda_i$ a non-zero eigenvalue $\lambda_i$ to $\vec{u}_i$. This map is by definition an isometry from $(\mathcal{S}^D_{\sim},\langle\cdot,\cdot\rangle_{\sim},\oplus_{\sim},\odot_{\sim})$ onto $(\mathbb{R}^{D-k},\langle \cdot,\cdot \rangle_2,+,\cdot)$. The inverse is then given by $\vec{z} \mapsto \sum_{i=1}^{D-M} z_i \frac{1}{\sqrt{\lambda_i}} \vec{u}_i $. A weighted $\clr$ map is given as defined in Subsection \ref{clrAndIlr}.

Note that in the case of compositional data, $\matr{L}_{\matr{W}} = \matr{L}_{\mathcal{A}}$, (\ref{equivIP}) is equivalent to the Aitchison inner product (\ref{aitchisonInner}) and (\ref{equivplus}) as well as (\ref{equivtimes}) 
is proportional to the scaled versions of perturbation and powering in the former.  
By definition of the quotient space any two elements are the same if their difference lies in the kernel. For compositional data this means $\log(\vec{x}) \equiv \log(\vec{x}) + \alpha \vec{1}$ for any $\alpha \in \mathbb{R}$. The usual condition $\sum_{j = 1}^{D} x_j = 1$ as used in the definition of $\mathcal{S}^D$, is simply a matter of fixing a representative for each equivalence class $[\log(\vec{x})]$,
but others could be chosen as well. The same goes for the graphical extension represented so far in this paper. 

The quotient space nature of regular compositional data has been investigated extensively in \citet{barcelo2001mathematical}. The graphical approach presented here is thus a natural extension. Note, however, that if we allow also for negative weights then, depending on the kernel of the Laplacian, we might allow for more than invariance under rescaling on subgraphs, as the latter might contain more than the constant vectors $\vec{1}_{\mathcal{V}_m}$. 



\begin{example}[compositional data - normal distribution]
In classical CoDa $\vec{x}$ is assumed to be normally distributed if for some fixed $\matr{V}$ the $\ilr$-transformed data follows a multivariate normal, i.e. $\ilr_{\matr{V}}(\vec{x}) \sim \mathcal{N}(\vec{0},\matr{\Sigma})$, where $\matr{\Sigma} \in \mathbb{R}^{(D-1) \times(D-1)}  $ is a positive definite matrix. $\ilr_{\matr{V}}(\vec{x})$ can be written as $\matr{V}'\matr{L}_{\mathcal{A}} \log(\vec{x})$, and so we conclude that $\log(\vec{x})$ follows a degenerate multivariate normal distribution with covariance matrix $(\matr{L}_{\mathcal{A}}\matr{V} \matr{\Sigma}^{-1} \matr{V}'\matr{L}_{\mathcal{A}})^{+}$, $\log(\vec{x}) \sim \mathcal{N}(\vec{0},(\matr{L}_{\mathcal{A}}\matr{V} \matr{\Sigma}^{-1} \matr{V}'\matr{L}_{\mathcal{A}})^{+}) $, see also the next section; 
the superscript $+$ indicates the Moore-Penrose inverse, \citet{ben2003generalized}.
The matrix $\matr{L}:=\matr{L}_{\mathcal{A}}\matr{V} \matr{\Sigma}^{-1} \matr{V}'\matr{L}_{\mathcal{A}}$ is symmetric and positive semi-definite. Furthermore, $\vec{1}$ is in its nullspace. Any symmetric matrix $\matr{L}$ with $\matr{L}\vec{1} = \vec{0}$ can be decomposed into $\matr{L} = \diag{(\matr{W} \vec{1})}-\matr{W}$ for a symmetric matrix $\matr{W}$, with possibly negative entries, $w_{ij} = -b_{ij}$ for $i\neq j$, and zero diagonal. Therefore
\begin{align*}
    \log(\vec{x})' \matr{L} \log(\vec{x}) &=   \log(\vec{x})'\diag{(\matr{W} \vec{1})}-\matr{W} \log(\vec{x}) \\
    &= \sum_{i=1}^D \log(x_i)^2 \sum_{j = 1}^D w_{ij} - \sum_{i,j} \log(x_i) \log(x_j) w_{ij} \\
    & = \frac{1}{2}\big( \sum_{i=1}^D \log(x_i)^2 \sum_{j = 1}^D w_{ij} + \sum_{j=1}^D \log(x_j)^2 \sum_{i = 1}^D w_{ij} \big) - \sum_{i,j} \log(x_i) \log(x_j) w_{ij} \\
    &= \frac{1}{2} \sum_{i,j = 1}^D w_{ij}( \log(x_i)^2+\log(x_j)^2-2 \log(x_i)\log(x_j)) \\
    &= \frac{1}{2} \sum_{i,j=1}^D \log\bigg(\frac{x_i}{x_j}\bigg)^2 w_{ij}.
\end{align*}
The weights of $\frac{1}{2} \sum_{i,j=1}^D \log\big(\frac{x_i}{x_j}\big)^2 w_{ij}$ can now be negative which means that $ \log(\vec{x})' \matr{L} \log(\vec{x}) $ becomes for a negative weight smaller the bigger its corresponding ratio gets. 
\end{example}

\section{The weights \texorpdfstring{$ w_{ij}$}{} }

We only consider the case of positive weights in this section. The choice of weights $w_{ij}$ depends on the  application in mind. For example, when one deals with variables with spatial dependence it can make sense to choose the weights and thus the graph according to the spatial position. In another setting, if expert-knowledge leaves us to believe that only certain pre-chosen ratios are relevant for the statistical analysis, we set the weights for the latter to one and all others to zero. 
If there is no knowledge of the ratios or if one is interested in the ratios which in relation to the data are important it seems appropriate to assume that the data follows a distribution  with (improper) density
\begin{align} \label{dataDens}
\frac{1}{\sqrt{(2\pi)^{D} \mid \alpha\matr{I} +\matr{L}_{\matr{W}}^+\mid _{+} }}\exp(- \frac{1}{2}\norm{\vec{x}}_{\matr{W},\alpha}^2 ),
\end{align}
where $\mid \cdot \mid_{+}$ denotes the pseudo determinant, see \citet{minka1998inferring}. 
 For $\alpha = 0$, (\ref{dataDens}) is understood as a density defined on a subspace of $\mathbb{R}^D_{+}$.
 For $\alpha>0$, $\mid \cdot \mid_{+}$ collapses with the usual determinant $\mid \cdot \mid$. 

In the context of graphical models, it has been shown for multivariate data $\vec{u} \sim \mathcal{N}(\vec{0},\matr{\Sigma})$, with positive definite $\matr{\Sigma}$, that the precision matrix $\matr{\Sigma}^{-1}$ reveals the graph structure in form of conditional independence \citet{lauritzen1996graphical}. In \citet{yuan2007model} and  \citet{friedman2008sparse}, a penalized  log-likelihood problem is solved to find an estimate  $\hat{\matr{\Sigma}}$
\begin{align} \label{graphLasso}
    \hat{\matr{\Sigma}}^{-1} := \argmin_{\substack{\matr{A}\in \mathbb{R}^{d \times d}\\ \matr{A}' = \matr{A}, \matr{A} \, p.d}} \log(\mid \matr{A} \mid)- \tr{\big(\matr{A} (N^{-1} \matr{U}' \matr{U})\big)} + \lambda \sum_{i,j = 1}^D \mid A_{ij} \mid ,
\end{align}
where p.d is short for positive definite, $\matr{U} \in \mathbb{R}^{N \times D}$ is the data matrix and $\lambda$ is a parameter controlling the sparsity of the entries of $A_{ij}$.
It is natural to extend (\ref{graphLasso}) to the compositional graph case for the data matrix $\matr{X} \in \mathbb{R}_{+}^{N \times D}$ sampled according to (\ref{dataDens}). If we allow $\alpha > 0$ we can solve 
\begin{align} 
    &\min_{\substack{ \alpha > 0, \, \matr{L} = \matr{L}' \in \mathbb{R}^{D \times D} }} \log(\mid \alpha \matr{I} + \matr{L} \mid) -\tr{\big(\big(\alpha \matr{I} + \matr{L}\big)  (N^{-1}\log(\matr{X}') \log(\matr{X})\big) \big)} + \lambda \sum_{\substack{i,j = 1 \\ i\neq j}}^D \mid L_{ij} \mid \label{probAbs1} \\
    & \qquad \qquad \qquad \text{ s.t. } \matr{L} \vec{1} = \vec{0}, \, \,  L_{ij} \leq 0 \text{ for }  \forall i \neq j \, \, \, \text{tr}(\matr{L})  = D-1 , \label{probAbs2}
\end{align}
where the $\log$ is applied coordinatewise and $\lambda$ is a fixed sparsity parameter. Problem (\ref{probAbs1})-(\ref{probAbs2}), rewritten only in terms of weights $L_{ij}$, was derived in \citet{lake2010discovering} in a non-compositional context. It is a convex problem which can be  efficiently solved. The condition $\text{tr}(\matr{L})  = D-1 $ ensures compatibility with the compositional case, $\matr{L}_{\mathcal{A}}$. The first two conditions mean that $\matr{L}$ can be decomposed as (\ref{mainform}).

If finding an estimator with $\alpha = 0$ is the goal, then as the pseudo-determinant is a discontinuous function, see \citet{holbrook2018differentiating}, this would lead to a discontinuous optimization problem that is hard to solve. 
For compositional data an approach for finding an underlying graph structure has been proposed in \citet{kurtz2015sparse}. After $\clr$ transforming the data, $\vec{t}_n := \clr{(\vec{x}_n)}$, for $n=1,\ldots,N$, and collecting the latter row-wise in the data matrix $\matr{T}$,
 one solves either (\ref{graphLasso}) with the new data $\vec{t}_n$, or a series of problems derived from the non-compositional setting in \citet{meinshausen2006high}:
\begin{align}  \label{MBmethod}
\min_{\vec{\beta} \in \mathbb{R}^{D-1}} \frac{1}{N}\norm{\vec{T}_{i} - \matr{T}^{-i}\vec{\beta}_i}_2^2    + \lambda \sum_{j=1}^{D-1} \mid (\vec{\beta}_i)_{j} \mid
\end{align}
for $i =1,\ldots,D$ and $\lambda \geq 0$, where $\vec{T}_{i}$ denotes the $i$-th column of the data matrix $\matr{T}$ and $\matr{T}^{-i}$ the latter with the $i$-th column deleted. 
To obtain a weight matrix a post-processing step is applied by setting each weight to $\Tilde{w}_{ij}:= $
$\mid \frac{1}{2}(\beta_{ij}+\beta_{ji})\mid$ for $i\neq j$ and $\Tilde{w}_{ii} := 0$ for $i=j$.
Various other methods for finding meaningful weights have been developed in the area of graph signal processing in a non-compositional context, see \citet{ dong2016learning}, \citet{kalofolias2016learn} and \citet{egilmez2017graph}. Following \citet{kalofolias2016learn}, a slightly reformulated compositional version of the 
proposed method would be
\begin{align} 
      \min_{ \substack{ \matr{W} = \matr{W}' \in \mathbb{R}^{D \times D} \\ \diag{\matr{W}} = \vec{0} \\ w_{ij}\geq 0 }} & \sum_{i,j = 1}^{D} \bigg(\frac{1}{N} \sum_{n=1}^N \log \bigg(\frac{x_{ni}}{x_{nj}}\bigg)^2\bigg) w_{ij} - \alpha \sum_{i=1}^D \log\bigg(\sum_{j=1}^D w_{ij}\bigg)+ \beta \sum_{i, j = 1}^{D} w_{ij}^2  \label{kalo} 
\end{align}
for two positive parameters $\alpha$ and $\beta$ controlling the sparsity of the resulting graph. 

The optimization problem (\ref{kalo}) leads to weights, $w_{ij}$, that are smaller the bigger the variance of the log-ratio $\log(\frac{x_i}{x_j})$ gets, and vice-versa. Something similar holds for the weights defined by (\ref{MBmethod}). This is the contrary of what we are interested in when defining an inner product and subsequently a norm that measures similarity as described above. We are rather interested in weights that put emphasis on log-ratios that explain a lot of variance, as we believe that those are of major interest. Log-ratios which have low variance should get small weights.
In the following we assume that each data point $\vec{x_n}$ has been replaced by $\clr(\vec{x_n})$ and that $\frac{1}{N}\sum_{n=1}^N\log(\vec{x_n}) = \vec{0}$ holds -- this can be achieved by centering each row and column of the log-transformed data matrix.
To get weights that put emphasis on log-ratios which are important in explaining the variance of the data set, we propose to use the same method as in \citet{greenacre2019variable} in a similar context with an additional step:
\begin{algorithm}
\caption{}
\label{alg:greenPlus}
\begin{algorithmic}[1]
\Statex \textbullet~\textbf{Start:} Compute the data matrix $\matr{Z} \in \mathbb{R}^{N\times \frac{D(D-1)}{2} }$ of all possible log-ratios. Denote $\Gamma$ the set of all possible edges, and set $\mathcal{E} = \emptyset$, $\matr{W}= \matr{0}$, $t = 0$ as well as $R_0 = 0$.
\State If $t < D$ go to 1, else stop.
\State For each $(k,l) \in \Gamma$ calculate the R-squared, $R^2(k,l)$, of the data points $\clr(\vec{x}_n)$ regressed onto the log-ratios $\{ \log(\frac{x_{ni}}{x_{nj}}) \mid (i,j) \in \mathcal{E} \cup \{(k,l)\} \}$.
\State Set $\mathcal{E}:=\mathcal{E} \cup (k_0,l_0)$
    for $(k_0,l_0) = \argmax_{(k,l) \in \Gamma} R^2(k,l)$,  $\Gamma = \Gamma \setminus \{(k_0,l_0)\} $, $R_{t}=R^2(k_0,l_0)$,  and $ w_{k_0,l_0} = w_{l_0,k_0} = R_t-R_{t-1}  $.
\State Set $t=t+1$ and go back to 0.
\end{algorithmic}
\end{algorithm}

\section{Two examples}

To illustrate the proposed framework we look at two real data examples. In both examples we will use Algorithm (\ref{alg:greenPlus}) to obtain the weights and a graph.

\subsection{Nugent data set}

In this first example we look at vaginal bacterial communities data. As this data set consist of microbiome data we will treat it as compositional, see \citet{gloor2017microbiome}, equivalently to \citet{lubbe2021comparison}. 
This data set consists
of 388 samples of 16s ribosomal RNA sequences recorded as 84 different operational taxonomic units (OTUs), see
\citet{ravel2011vaginal}. As a preprocessing step only OTUs present in at least 5\% of the samples were kept and subsequently
the 82\% of zeros of the remaining resulting data matrix $\matr{X} \in \mathbb{R}^{388 \times 84}_+$ were replaced by a value of 0.5, compare to \citet{lubbe2021comparison}.
Additionally to $\matr{X} $, a response $\vec{y} \in \{0,\ldots,10\}^{388}$ of 11 categories was recorded indicating the 
risk of bacterial vaginosis, from low risk, given by 0, to high risk, given by 10.

We scale $\vec{y}$ and center row- as well as column-wise the log-transformed data matrix $\log(\matr{X} )$.
To obtain a graph, the stepwise algorithm as explained in Algorithm~\ref{alg:greenPlus} is employed. This leads to a sequence
of $D-1$ edges $\mathcal{E} = \{e_1 =(i_1,j_1),\ldots,e_{D-1}=(i_{D-1},j_{D-1}) \}$, corresponding weights $w_{e_1},\ldots,w_{e_{D-1}}$, and a 
Laplacian matrix $\matr{L}_{\matr{W}}$. 

Figure \ref{fig:cumVarEx3} compares the cumulative explained variance of principal components (PCs) in the classical compositional case, see for example \citet{aitchison1983principal} or  \citet{filzmoser2009principal}, the continuous black line,
to the cumulative explained variance of the log-ratios belonging to the edge set $E$, dotted black line, as well as to
the projected data, by (\ref{ilr2}) -- after possible reordering, dashed black line.
All three methods yield very similar results. The PC directions by definition explain the most cumulative variance.
The dashed as well as the dotted line, which overlap in this example, are very close to
the continuous black line, which indicates that this choice of weights $w_{e_1},\ldots,w_{e_{D-1}}$
only leads to a marginal loss of information. 
Instead of using simply log-ratios corresponding to the edge set $\mathcal{E}$,
the advantage
by defining a Laplacian matrix $\matr{L}_{\matr{W}}$ is that a Hilbert space is constructed.

\begin{figure}[t]
    \centering
    \includegraphics[scale = 0.50]{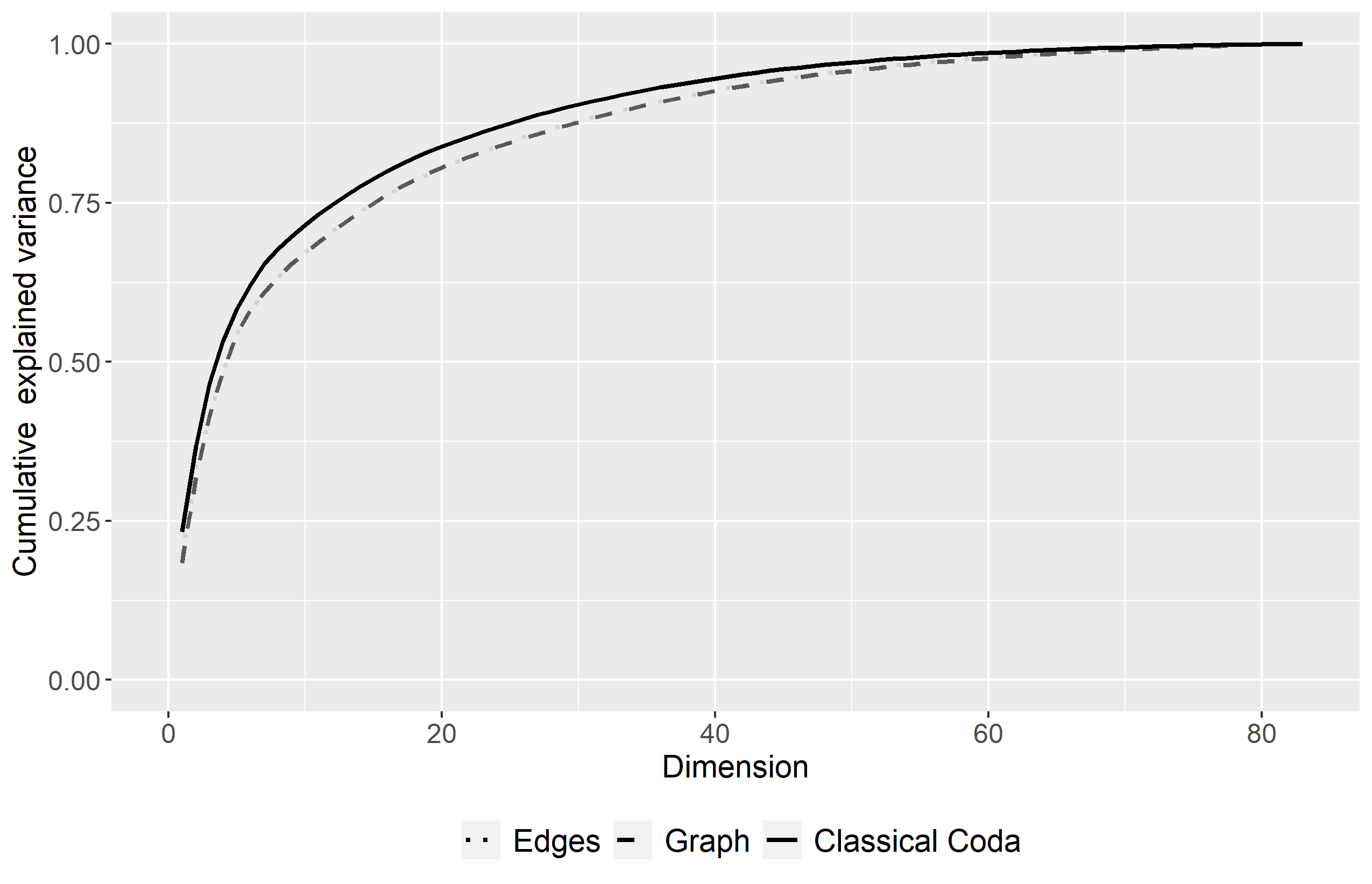}
    \caption{Plot of the cumulative explained variances. The black line shows the cumulative explained variance of regular compositional PCA dependent on the number of components. The dashed/dotted ones, overlapping, show the 
    R-squared for each step of Algorithm \ref{alg:greenPlus} as well as the 
    R-squared for the coordinates of map (\ref{ilr2})  dependent on the number of such coordinates, after possible reordering.}
    \label{fig:cumVarEx3}
\end{figure}

The framework presented in this paper can be used to get different interpretations for
a regression model in CoDa. The idea is to first fit a regression model in the
Aitchison geometry to the data $(\matr{x},\vec{y})$, i.e. 
$y_n =  \langle \vec{a},\vec{x}_n \rangle_{A} + \epsilon_n$, for $n = 1,...,388$, with $\epsilon_n \sim \mathcal{N}(0,\sigma^2)$
and then
find a coefficient vector $\vec{b}$ such that $\langle \vec{a},\vec{x} \rangle_{A} \approx
\langle \vec{b},\vec{x} \rangle_{\matr{W}} $. The coefficient vector $\vec{b}$ is chosen as a solution of  $ \min_{\vec{b}} \norm{ \matr{L}_{A} \log(\vec{a}) -  \matr{L}_{\matr{W}} \log(\vec{b}) }^2_2 $, so
 $ \vec{b} = \exp(\matr{L}_{\matr{W}}^+ \matr{L}_{A} \log(\vec{a}) ) $.
For any fixed $k =1,\ldots,D-1$ we can compute the coefficient vector $\vec{b}^{k}:=\exp(\matr{L}_{\matr{W}^k}^+ \matr{L}_{A} \log(\vec{a})) $ where 
$\matr{W}^k $ is the weight matrix induced by the weights $w_{e_1},\ldots,w_{e_{k}}$ leading to $D-1$ different 
linear models $f_k(\vec{x}):=\langle \vec{b}^k,\vec{x} \rangle_{\matr{W}^k}$. As $ f_k(\vec{x}) = 
\sum_{i,j = 1}^D  \log(\frac{x_i}{x_j})\log(\frac{b^k_i}{b^k_j}) w^k_{ij} = \sum_{i < j}  \log(\frac{x_i}{x_j}) 2 \log(\frac{b^k_i}{b^k_j})w^k_{ij}$
 we can visualize for a fixed $k$ this linear model, which can be thought of as a signed weighted sum of log-ratios, as a graph where the edges are 
induced by $\widetilde{w}_{ij}^k := 2 \log(\frac{b^k_i}{b^k_j}) w^k_{ij} \neq 0$.
To obtain a $k$ that is appropriate we split for each fixed $k$ the data set $(\matr{X},\vec{y})$ a hundred times randomly 
into a training and a test set, compute on the training set the zeroSum solution $\vec{a}$, see \citet{lin2014variable}, 
compute $\vec{b}^k$ and calculate the mean squared error (MSE) 
for the zeroSum model and $f_k(\vec{x})$ on the test set. We use a zeroSum model, as in \citet{lubbe2021comparison}, to get a sparse coefficient vector $\vec{a}$, only retaining important variables for explaining the response - which ultimately leads to a lower mean MSE than the full one.
Figure~\ref{fig:mseplotex3} shows for each $k$ and both methods the mean of the MSE over all hundred repetitions.
\begin{figure}[t]
    \centering
    \includegraphics[scale = 0.45]{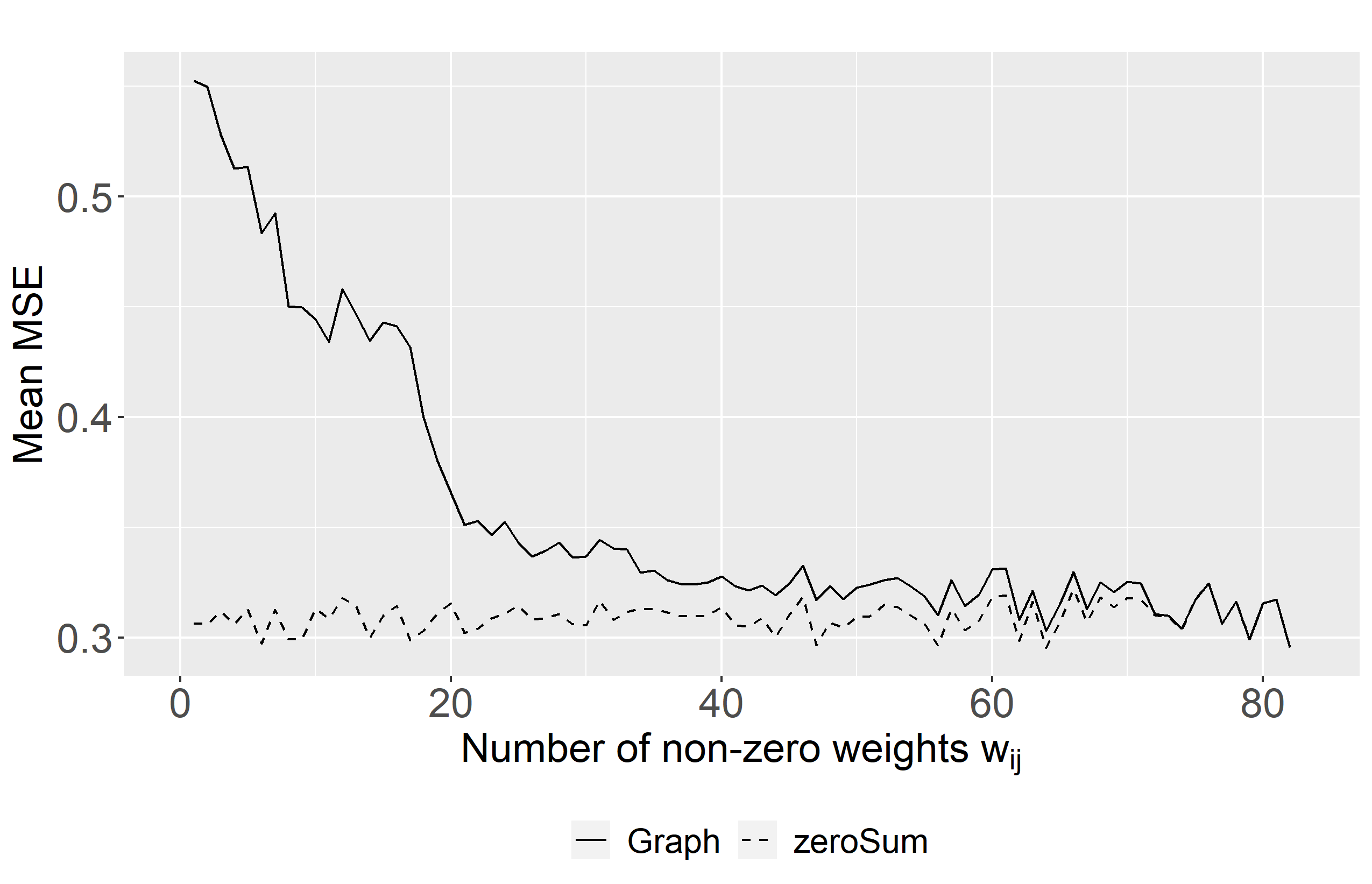}
    \caption{Plot of the mean of the mean squared errors (MSEs) comparing the regression model of the graph method to the zeroSum model as a baseline over one hundred repetitions each. For the graph method the horizontal axis corresponds to the non-zero weights found by Algorithm~\ref{alg:greenPlus} in each step.}
    \label{fig:mseplotex3}
\end{figure}

The value $ k = 25$ seems to be a good trade-off between low MSE in median and sparse graph. For the latter we show in Figure \ref{fig:ex3regplot}
\begin{figure}[t]
    \centering
    \includegraphics[scale = 0.35]{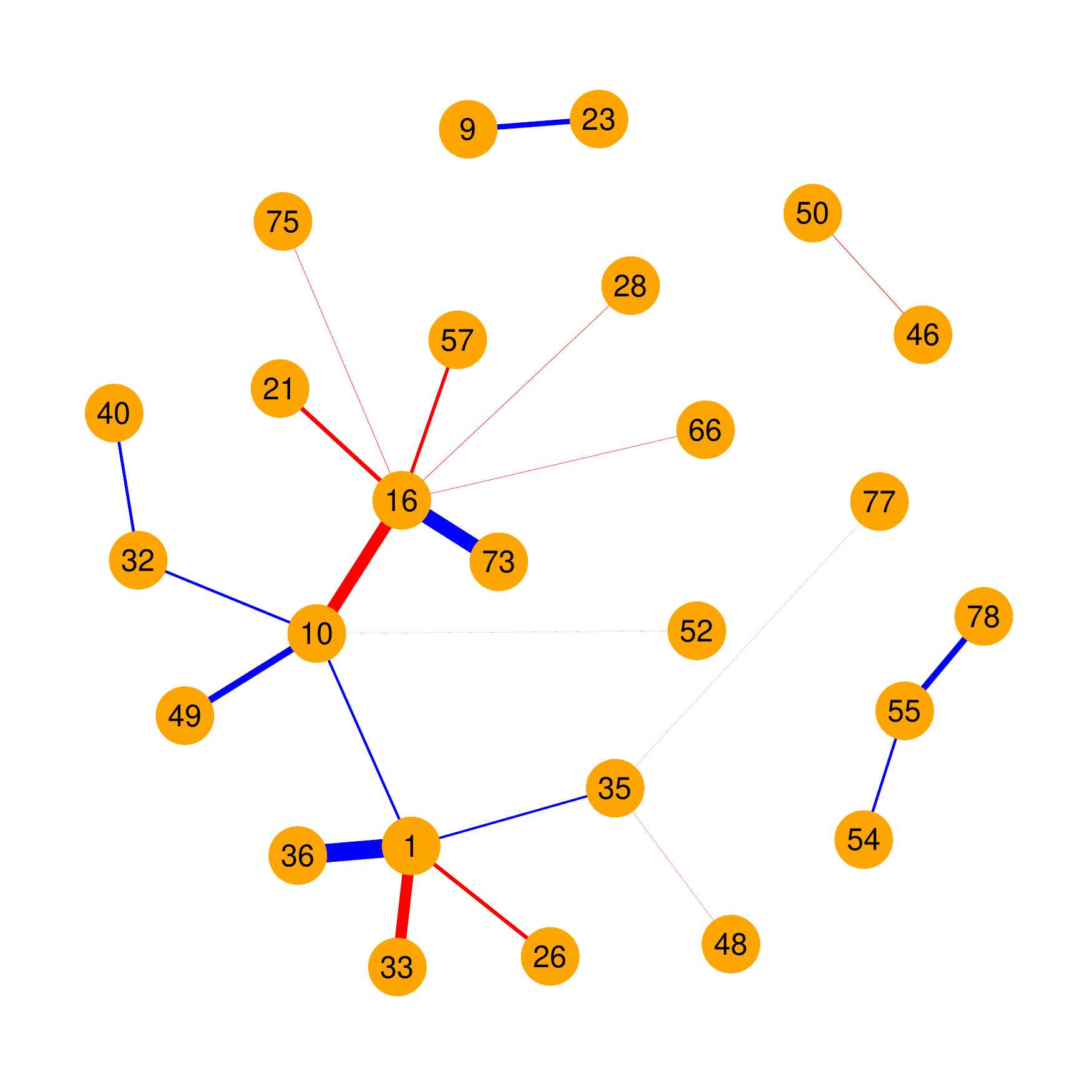}
    \caption{Visualization of the chosen regression model, $k=25$, $\sum_{i<j}^D \log\big(\frac{x_i}{x_j}\big) 2 \log\big(\frac{b^k_i}{b^k_j}\big)w^k_{ij}$. The edge width is proportional to the size of $\sigma_{ij} 2 \log \big(\frac{b^k_i}{b^k_j}\big)w^k_{ij}$.
    It is red when the latter is positive and blue otherwise.}
    \label{fig:ex3regplot}
\end{figure}
 the graph induced by 
$\widetilde{w}_{ij}^k := 2 \log(\frac{b^k_i}{b^k_j}) w^k_{ij} $.  
The edge thickness is proportional to 
$ \sigma_{ij} \mid \widetilde{w}_{ij}^k \mid$,
where $\sigma_{ij}$ is the standard deviation of $\log(\frac{x_i}{x_j})$;
red edges indicate $\widetilde{w}_{ij}^k > 0$, blue ones $\widetilde{w}_{ij}^k < 0$.
We can see that there are four non-connected sub-graphs. Table \ref{table:correspondingNbrs} shows the phylotypes corresponding to the numbers in Figure \ref{fig:ex3regplot}. In the latter we can see, for example, that Lactobacillus\_crispatus (16) and Lactobacillus\_iners (1) are central phylotypes in the graph with many connections to others,
and thus they seem to have a big effect in governing the risk of a disease. This is in accordance with findings in \citet{ravel2011vaginal}, where higher proportions of phylotypes
of Lactobacillus were associated with lower risk. We can see in Figure \ref{fig:ex3regplot} that the connections between
Lactobacillus\_crispatus (16) and  Atopobium\_vaginae (73) as well as between Lactobacillus\_iners (1) and Prevotella\_timonensis (36)
are comparably big. Atopobium\_vaginae (73) and  Prevotella\_timonensis (36) are two of multiple phylotypes associated with higher risk, see \citet{ravel2011vaginal}, which 
is in agreement with the display in Figure \ref{fig:ex3regplot} as a blue edge indicating a negative coefficient. Note that the log-ratio between Ureaplasma\_parvum\_serovar (10) and Lactobacillus\_crispatus (16) seems to have a big influence on the risk as well. A smaller effect on the risk can be seen in Figure \ref{fig:ex3regplot}, for example, for the log-ratio between Anaerococcus\_lactolyticus (55) and Gardnerella\_vaginalis (78).

\begin{table}[ht]
\centering
\begin{tabular}{|r|l|l|l|}
  \hline
 1 & Lactobacillus\_iners & 46 & Anaerococcus\_hydrogenalis \\ 
   9 & Veillonella\_montpellierensis & 48 & Prevotella\_bivia \\ 
   10 & Ureaplasma\_parvum\_serovar & 49 & Prevotella\_buccalis \\ 
   16 & Lactobacillus\_crispatus & 50 & Prevotella\_disiens \\ 
   21 & Lactobacillus\_vaginalis & 52 & Finegoldia\_magna \\ 
   23 & Dialister\_micraerophilus & 54 & Lactobacillus\_acidophilus \\ 
   26 & Lactobacillus\_gasseri & 55 & Anaerococcus\_lactolyticus \\ 
   28 & Anaerococcus\_tetradius & 57 & Lactobacillus\_jensenii \\ 
   32 & Peptostreptococcus\_anaerobius & 66 & Streptococcus\_agalactiae \\ 
   33 & Corynebacterium\_tuscaniense & 73 & Atopobium\_vaginae \\ 
  35 & Prevotella\_amnii & 75 & Anaerococcus\_obesiensis \\ 
   36 & Prevotella\_timonensis & 77 & Fenollaria\_massiliensis \\ 
   40 & Aerococcus\_christensenii & 78 & Gardnerella\_vaginalis \\ 
   \hline
\end{tabular}

\vspace*{2mm}
\caption{Corresponding phylotypes to the vertex numbers in Figure \ref{fig:ex3regplot}}
\label{table:correspondingNbrs}
\end{table}

\subsection{Kola data set}

As a second example we analyze the Kola data set of the moss layer which is available in the R package mvoutlier \citet{mvoutlier}. It is comprised of 31 chemical concentrations of the elements Ag, Al, As, B, Ba, Bi, Ca, Cd, Co, Cr, Cu, Fe, Hg, K, Mg, Mn, Mo, Na, Ni, P, Pb, Rb, S, Sb, Si, Sr, Th, Tl, U, V and Zn, sampled at 598 different locations in the the west of the Kola peninsula and the northernmost part of Finland and Norway, see \citet{reimannenvironmental}. 
All data are collected in the data matrix $\matr{X} \in  \mathbb{R}^{598 \times 31}$ which we center row- and column-wise after log-transforming each entry. As in the previous example we use  Algorithm~\ref{alg:greenPlus} to find weights $w_{ij}$. In Figure~\ref{fig:ex3CoordMaps} we display three selected coordinates of the map (\ref{ilr2}) on the left, and on the right the corresponding weights $\vec{u}_i$. The most upper plot on the right indicates that the according coordinate displayed on its left is influenced mainly by two elements, Nickel (Ni) and Sulfur (S), which are also connected in the graph.
We can see high values around the cities of  Monchegorsk and Nikel/Zapolyarny tapering off radially from the cities. 
At both locations there is a huge industry with nickel refining facilities, and the moss layer collects these emissions, which seem to be highly visible in the ratio of Ni and S. 
The middle row of Figure~\ref{fig:ex3CoordMaps} displays high values around the area of the Khibiny Massif with an anomaly in Norway. The elements involved are Sodium (Na), Strontium (Sr), Titanium (Ti) and to some extent Boron (B) and Chromium (Cr). Supposedly this map can be related to alkaline intrusion effects.
The right plot in the last row of Figure~\ref{fig:ex3CoordMaps} shows that the corresponding coordinate displayed on its left is mainly highly influenced by effects of the elements  Vanadium (V), Magnesium	(Mg), Aluminium	(Al), Arsenic	(As) and Molybdenum (Mo). Relatively high values of this coordinate can be seen in Russia,
along the main road connections from Murmansk on the coast, down to Monchegorsk and Apatity, and towards the West to Kovdor. Elements such as V and Al could typically indicate pollution from dust, probably caused by ore transportation.

\begin{figure}[t]
    \centering
    \hspace*{-1.5cm} 
    \includegraphics[scale = 0.8]{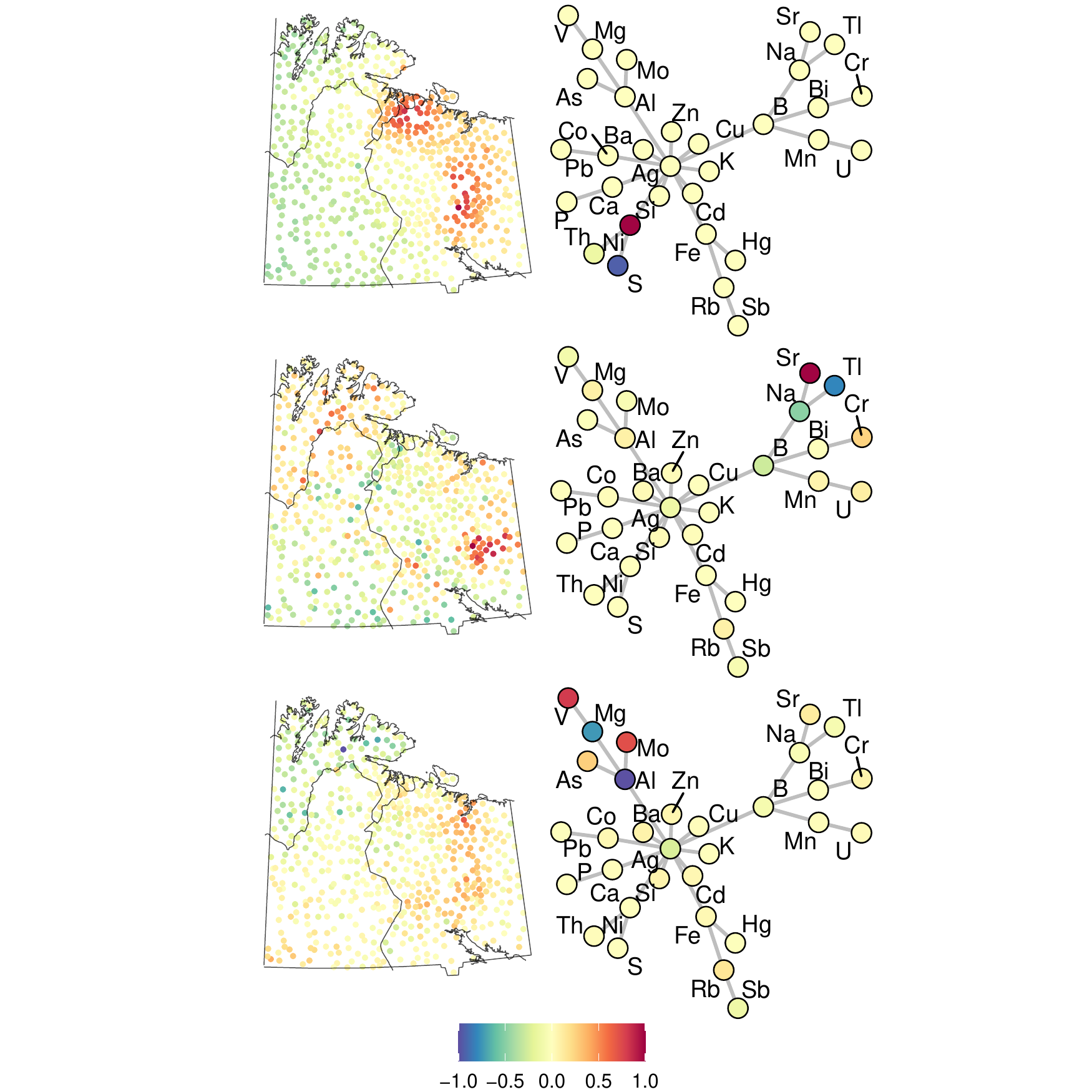}
    \caption{The left column shows the scaled GILR2-map (\ref{ilr2}) such that its values lie in the interval $[-1,1]$. The right column shows, again scaled to the $[-1,1]$ interval, in each node the corresponding entry of the vector $\vec{u}_i$ of (\ref{ilr2}). The underlying graph structure was found by Algorithm~\ref{alg:greenPlus}.}
    \label{fig:ex3CoordMaps}
\end{figure}

\section{Conclusions}

This paper is an attempt to link compositional data analysis with the concepts of signal processing on graphs.
We started from the observation that the Aitchison norm is equally influenced by
all pairwise log-ratios between the variables, regardless if these are meaningful from the problem context or not. Modifying the Aitchison norm by putting a non-negative weight onto each (squared) pairwise log-ratio led to a norm that gives different log-ratios a different impact. With this change comes a semi-inner product induced by the Laplacian matrix, well known in graph theory, and a different geometry that still satisfies the most important properties for compositional data analysis, such as scale invariance and compositional coherence. This radically differs from previous approaches that did consider weighting, for example by constructing an $\ilr$ map with the first coordinate being a weighted sum of log-ratios, and the others being such that they form an appropriate basis, without changing the underlying (Aitchison) geometry. 

The framework we propose is very flexible and it includes many extensions that allow for different additional modeling choices, such as accounting for absolute effects of variables to only considering a low number of interesting subsets of weighted balances -- that is, coordinates that consist of a few variables in the numerator as well as denominator. We could show that such modeling choices do not lead to any loss in interpretability and that in fact mappings, such as $\clr$ and $\ilr$, which are of central importance in compositional data analysis, have an analogue in the graph setting.
To find appropriate weights we resorted to a stepwise algorithm that selects log-ratios and their weighting, in a sequential manner, based on explained variance of the whole data set.

To show its utility, we applied the proposed methodology to real data sets. In the first example we looked at a regression problem by modeling the risk of getting bacterial vaginosis as being depended on different  phylotypes. We showed that the graph perspective can be used to gain different interpretable insights by first reducing the dimension and then visualizing the resulting model as a graph. The risk was mainly driven by only a few central phylotypes having strong connections to others. 
In the second example, we looked at the Kola data set which consists of geographically dependent chemical compositions. Choosing the weights in a data dependent way allowed us to construct different interpretable orthogonal maps as given by the theory developed in this paper. The effect of these maps could be visualized as graphs with nodes colored according to the information of each element they carried. 

In the future we intend to look into the performance and possible extensions of the methods presented in this paper to different settings such as classification or regression. Also, we plan to investigate further the problem of finding weights under different constructural constraints and compare these methods with already existing ones in the signal processing field. Additionally, it would also be interesting to investigate deviations from the assumptions in (\ref{dataDens}) to more robust cases.

\section*{Statements and Declarations}

The authors declare that they have no conflict of interest.

\begin{appendices}

\section{Proofs}

\begin{proof}[Proof of Lemma \ref{graphCodaInnerProd}]

The logarithm maps $(\mathbb{R}^D_+,\oplus,\odot)$ one-to-one into $(\mathbb{R}^D,+,\cdot)$, which means that $(\mathbb{R}^D_+,\oplus,\odot)$ is also a vector space. Linearity in each argument and conjugate symmetry of $ \langle \cdot, \cdot \rangle_{\matr{W},\alpha}$ follows from linearity of $\matr{d}_{\matr{W}^{\frac{1}{2}}}$ resp. $\matr{d}_{\matr{W}^{\frac{1}{2}}}'$ and the linearity of the standard inner product on $\mathbb{R}^D$ and $ \mathbb{R}^{\mid E \mid}$. $ \langle \vec{x}, \vec{x} \rangle_{\matr{W},\alpha} \geq 0 $ holds by definition. When $ \langle \vec{x}, \vec{x} \rangle_{\matr{W},\alpha} = 0$, since $\matr{L}_{\matr{W}}$ is positive semi-definite, we can conclude $\alpha \langle \log(\vec{x}), \log(\vec{x})  \rangle_{2}=0$ and therefore $\log(\vec{x}) = 0$, $\vec{x} = \vec{1}$.

As $(\mathcal{S}^D_{\matr{W}},\oplus_{\matr{W}},\odot_{\matr{W}})$ can be seen as $M$ separate D-part simplices, $\mathcal{S}^{m}$, $m = 1,\ldots,M$ - not taking into account any inner products for the moment -  we conclude that $(\mathcal{S}^D_{\matr{W}},\oplus_{\matr{W}},\odot_{\matr{W}})$ is also a vector space. 
What is left is to show is that from $ \langle \vec{x}, \vec{x} \rangle_{\matr{W},0} = 0$ we can conclude $\vec{x} = \sum_{m = 1}^M \frac{\kappa_m}{\mid \mathcal{V}_m \mid}\vec{1}_{i \in \mathcal{V}_m}$. Using the eigenvalue decomposition of $\matr{L}_{\matr{W}} = \matr{U} \Sigma \matr{U}'$ we can write, 
$ \langle \vec{x}, \vec{x} \rangle_{\matr{W},0} =   \log(\vec{x})'\matr{U}  \Sigma  \matr{U}'\log(\vec{x}) 
= \sum_{j = 1}^D \lambda_j (\langle \vec{u}_j ,\log(\vec{x}) \rangle)^2$, with $\lambda_j \geq 0$ and $\lambda_j = 0$ iff an eigenvalue of $\matr{L}_{\matr{W}} $ is zero. So $\langle \vec{u}_j ,\log(\vec{x}) \rangle$ must be zero whenever $\lambda_j >0$ which is equivalent to saying that $\log(\vec{x})$ is in the kernel. We know from Lemma \ref{LaplacianLemma} that the kernel of $\matr{L}_{\matr{W}}$ is spanned by $\vec{1}_{i \in \mathcal{V}_m}$, so there exist $c_1,\ldots,c_m \in \mathbb{R}$ s.t. $\log(\vec{x}) = \sum_{m=1}^M c_m \vec{1}_{i \in \mathcal{V}_m}$. From the latter we get $\log(\vec{x}_{\mathcal{V}_m}) = c_m \vec{1}$, where $\vec{x}_{\mathcal{V}_m}$ denotes the entries of $\vec{x}$ with index in $\mathcal{V}_m$, and thus $\vec{x}_{\mathcal{V}_m} = \exp(c_m) \vec{1}$. As, by definition, $\sum_{i \in \mathcal{V}_m} x_i = \kappa_m$ holds, we conclude $\vec{x}_{\mathcal{V}_m} =\frac{\kappa_m}{\mid \mathcal{V}_m \mid} \vec{1}$, $\vec{x} = \sum_{m = 1}^M \frac{\kappa_m}{ \mid \mathcal{V}_m \mid}\vec{1}_{i \in \mathcal{V}_m}$.
\end{proof}

\begin{proof}[Proof that  $\norm{ \cdot }_{q,\alpha}$ is a norm on $(\mathbb{R}^D_+,\oplus,\odot)$ resp. $(\mathcal{S}^D_{\matr{W}},\oplus_{\matr{W}},\odot_{\matr{W}})$]
As discussed in the proof of Lemma \ref{graphCodaInnerProd} both $(\mathbb{R}^D_+,\oplus,\odot)$ resp. $(\mathcal{S}^D_{\matr{W}},\oplus_{\matr{W}},\odot_{\matr{W}})$ are vector spaces. As the standard q-norm $\norm{\cdot}_q$ is a norm on $\mathbb{R}^D$ all conditions for $\norm{ \cdot }_{q,\alpha}$ being a norm, but one are trivially fulfilled, as $\vec{x} \mapsto \log(\vec{x})$ maps into a subset of $\mathbb{R}^D$.
We only need to proof that from $\norm{\vec{x}}_{q,\alpha} = 0$ we can deduce that $\vec{x}$ is the neutral element. 
For $\alpha \neq 0 $ from $\norm{\vec{x}}_{q,\alpha} = 0$ we get $\norm{\log(\vec{x})}_q = 0$ and so $\vec{x} = \vec{1}$. For $\alpha_= 0$ we get from $\norm{\vec{x}}_{q,0} = 0$, $\matr{d} \log(\vec{x}) = 0$, where we dropped the subscript. For each connected subgraph $\matr{d} \log(\vec{x}) = 0$ is equal to $\vec{x}_{i \in  \mathcal{V}_m} = c_m \vec{1}$. Again we can conclude as $\sum_{i \in \mathcal{V}_m} x_i = \kappa_m$ holds, $\vec{x} = \sum_{m = 1}^M \frac{\kappa_m}{\mid \mathcal{V}_m \mid}\vec{1}_{i \in \mathcal{V}_m}$.
\end{proof}

\begin{proof}[Proof of Lemma \ref{cholForgraph}]
First we consider the case where the graph is connected. By Lemma \ref{graphCodaInnerProd} we know that $\matr{L}_{\matr{W}}\vec{1} = 0$. Therefore, as $0= \matr{L}_{\matr{W}}\vec{1} = \matr{U} \matr{\Sigma} \matr{U}' \vec{1} = \sum_{j = 1}^{D-1}  \lambda_j(\sum_{i=1}^D u_{ij})\vec{u}_{j} $, where $\vec{u}_j$ is the j-th column of $\vec{U}$ and orthogonal to each other, we deduce $\sum_{i=1}^D u_{ij} = 0$ for $j = 1,...,D-1$. 
Compute the QR decomposition of $\matr{U} \matr{\Sigma}^{\frac{1}{2}} \matr{U}'$, $ \matr{U}\matr{\Sigma}^{\frac{1}{2}} \matr{U}' = \matr{Q} \matr{R}$. With this $\matr{U} \matr{\Sigma}  \matr{U}' = (\matr{U}\matr{\Sigma}^{\frac{1}{2}} \matr{U}')'(\matr{U}\matr{\Sigma}^{\frac{1}{2}} \matr{U}') = \matr{R}'\matr{Q}'\matr{Q}\matr{R}= \matr{R}'\matr{R}$ holds. As we also have $ \matr{Q}\matr{R}\vec{1} = \matr{U} \matr{\Sigma}^{\frac{1}{2}} \matr{U}' \vec{1} = \vec{0}$, from $\sum_{i=1}^D u_{ij} = 0$ for $j=1,\ldots,D-1$, we see that $\matr{R}\vec{1} = \vec{0}$ must hold. From the latter we immediately obtain $R_{DD}  = 0$ as $\matr{R}$ is an upper triangular matrix. As $\matr{L}_{\matr{W}}$ has only one non-zero eigenvalue all other diagonal elements of $\matr{R}$ must be non-zero  due to the ranks of both matrices needed to be equal. As the sign can be chosen, we are done by writing $\matr{C}$ instead of $\matr{R}$. The non-connected setting follows directly from looking at each $\matr{L}_m$ individually.
\end{proof}

\begin{proof}[Proof Lemma \ref{graphILR}]
Write $\matr{P} = \diag{(\matr{P_1},\ldots,\matr{P}_M)}$.
It is easy to see that both maps are linear mappings. 
The map (\ref{ilrReduced}) has rank $D-M$ as all but one diagonal element of each $\matr{C}_m$ are positive, by Lemma \ref{cholForgraph}; meaning that each $\matr{C}_m$ has rank $\mid \mathcal{V} \mid_m - 1$. This means that the map (\ref{ilrReduced}) is surjective. What remains to show is that it is isometric as injectivity then directly follows. 
We have 
\begin{align*}
    \norm{\vec{x}}^2_{\matr{W},0} &= \log(\vec{x})' \matr{L}_{\matr{W}} \log(\vec{x}) \\
    &= (\matr{P}\log(\vec{x}))'(\matr{P} \matr{L}_{\matr{W}} \matr{P}')(\matr{P} \log(\vec{x})) \\
    &= (\matr{P}\log(\vec{x}))'( \matr{L}_{\matr{P}\matr{W}\matr{P}'})(\matr{P} \log(\vec{x})) \\
    &= \sum_{m= 1}^M (\matr{P_m}\log(\vec{x}_{i \in \mathcal{V}_m}))' \matr{C}_m' \matr{C}_m (\matr{P_m}\log(\vec{x}_{i \in \mathcal{V}_m}))\\ &= \sum_{m= 1}^M (\matr{C}_m \matr{P_m}\log(\vec{x}_{i \in \mathcal{V}_m}))'  (\matr{C}_m \matr{P_m}\log(\vec{x}_{i \in \mathcal{V}_m}))\\
    &= \sum_{m= 1}^M (\matr{C}_{-\mid \mathcal{V}_m \mid} \matr{P_m}\log(\vec{x}_{i \in \mathcal{V}_m}))'  (\matr{C}_{-\mid \mathcal{V}_m \mid} \matr{P_m}\log(\vec{x}_{i \in \mathcal{V}_m})) \\
    &=  \norm{\diag{(\matr{C}_{-\mid \mathcal{V}_1 \mid},\ldots,\matr{C}_{-\mid \mathcal{V}_M \mid})} \diag{(\matr{P}_1,\ldots,\matr{P}_M)}  \log(\vec{x})}_2^2 
\end{align*}
which shows that  (\ref{ilrReduced}) is an isometry. As $\alpha \matr{I} +  \diag{(\matr{P}_1\matr{L}_{1}\matr{P}_1 ',\ldots,\matr{P}_M\matr{L}_{M}\matr{P}_M ')}$ has full rank for $\alpha \neq 0$, the map (\ref{ilrAll}) is surjective. The following chain of equations 
\begin{align*}
    \norm{\vec{x}}^2_{\matr{W},\alpha} &= 
    (\matr{P}\log(\vec{x}))'\matr{P} (\alpha \matr{I} + \alpha_2\matr{L}_{\matr{W}}) \matr{P}'\matr{P}\log(\vec{x})\\
    &=(\matr{P}\log(\vec{x}))'(\alpha \matr{I} + \alpha_2\matr{P}\matr{L}_{\matr{W}}\matr{P}') (\matr{P}\log(\vec{x}))\\
    &= \sum_{m = 1}^M(\matr{P}_m\log(\vec{x}_{i \in \mathcal{V}_m}))'(\alpha\matr{I} + \matr{P}_m\matr{L}_m\matr{P}_m') (\matr{P}_m\log(\vec{x}_{i \in \mathcal{V}_m}))\\
    &= \sum_{m = 1}^M(\matr{P}_m\log(\vec{x}_{i \in \mathcal{V}_m}))' \Tilde{\matr{C}}_m'\Tilde{\matr{C}}_m (\matr{P}_m\log(\vec{x}_{i \in \mathcal{V}_m}))\\
    &=\norm{(\diag{(\Tilde{\matr{C}}_1,\ldots,\Tilde{\matr{C}}_M)}\diag{(\matr{P}_1,\ldots,\matr{P}_M)}\log(\vec{x})}_2^2
\end{align*}
shows that it is an isometry.
\end{proof}

\end{appendices}



\end{document}